\documentclass[12pt, draftclsnofoot, onecolumn]{IEEEtran}
\usepackage{pgfplots}
\pgfplotsset{compat=newest}
\usepackage{epsfig}
\usepackage{amsthm,amssymb,graphicx,graphicx,multirow,amsmath,color}
\usepackage{caption}
\usetikzlibrary{plotmarks}
\usetikzlibrary{patterns}
\usetikzlibrary{calc,positioning,fit,backgrounds}
\usetikzlibrary{shapes,snakes}
\usetikzlibrary{intersections,positioning}
\usepackage{subcaption}
\usepackage[noadjust]{cite}
\def\argmin{\operatorname{arg~min}}

\def\E{\mathbb{E}}
\def\Y{\mathbf{Y}}
\def\H{\mathbf{H}}
\def\B{\mathbf{B}}

\def\Eb{\mathbf{E}}
\def\Lb{\mathbf{L}}
\def\X{\mathbf{X}}

\def\W{\mathbf{W}}
\def\u{\upsilon}

\def\P{\mathbb{P}}

\def\ie{{\em i.e.}}

\def\R{\mathbb{R}}
\def\No{\mathbf{N}}

\def\L{\mathcal{L}}

\def\i{\mathbf{1}}

\def\us{\mathbf{u}}
\def\B{\mathcal{B}}

\def\d{\mathrm{d}}

\def\p{\mathtt{P}}
\def\x{\mathtt{x}}

\def\I{\mathbf{I}}

\def\Q{\mathbf{Q}}

\def\sinr{\mathtt{SINR}}
\def\snr{\mathtt{SNR}}
\def\sir{\mathtt{SIR}}

\newtheorem{lemma}{Lemma}{}

\author{Sreejith~T.~Veetil,
        Kiran~Kuchi  and~ Radha Krishna Ganti
\thanks{Sreejith T. V. and Kiran Kuchi are  with the Department
of Electrical Engineering, IIT Hyderabad, India. Radha Krishna Ganti is  with the Department
of Electrical Engineering, IIT Madras, India.}
}

\title{Performance of Cloud Radio Networks}

\begin{document}
\maketitle
\begin{abstract}
Cloud radio networks coordinate transmission among base stations (BSs) to reduce the interference effects, particularly for the cell-edge users. In this paper, we analyze the performance of a cloud network with static clustering where geographically close BSs form a cloud network of cooperating BSs. Because, of finite cooperation, the interference in a practical cloud radio cannot be removed and in this paper, the distance based interference is taken into account in the analysis. In particular, we consider centralized zero forcing  equalizer and dirty paper precoding for cancelling the interference. Bounds are developed on the signal-to-interference ratio distribution and achievable rate with full and limited channel feedback from the cluster users. The adverse effect of finite clusters on the achievable rate is quantified. We show that, the number of cooperating BSs is more crucial than the cluster area  when full channel state information form the cluster is available for precoding. Also, we study the impact of limiting the channel state information on the achievable rate. We show that even with a practically feasible feedback of about five to six channel states from each user, significant gain in mean rate and cell edge rate compared to conventional cellular systems can be obtained. 
\end{abstract}

\section{Introduction}
The evolution of cellular communication networks from 1G through 4G \cite{3p9G:standard} has resulted in a steady increase in the allowable rates for users (UEs). In cellular systems employing universal frequency reuse, other cell interference (OCI) is the main bottleneck that limits the system capacity. Mitigation of OCI using interference cancellation strategies results in improvement of achievable  rates. Cloud radio network, which is also labeled as network multiple-input-multiple-output (MIMO) \cite{Telatar99capacityof,Caire:2009A,LTEcomp,Gorokhov,Fettweis,Gaal:2012,Foschini:2006a,Gersbert:Dec2010,Wang:Nov2011,Poor:2012a,Lozano:2009e} is a centralized precoding architecture that can effectively remove OCI. 

In a cloud radio network, a group of geographically close base-stations (BSs) are connected to a central processing unit (or cloud) through optical fiber. The baseband processing is done at the cloud, while the waveforms are exchanged between BSs and cloud through fiber. The function of the BS mainly involves carrier power amplification and transmission through an antenna. This type of BS with reduced functionality is referred to as a remote radio head-end (RRH). Use of low power RRHs generally reduces the size and cost significantly compared to conventional BSs. Furthermore, availability of baseband signals at the cloud enables interference suppression techniques to be used both in the downlink and the uplink. In addition, cloud facilitates traffic load balancing through joint scheduling among the RRHs. 

In this paper, we are primarily interested in determining the rate gain that can be obtained through joint processing of signals  with geographical clustering and limited channel feedback.

\subsection{Motivation and Related Works}

In \cite{FoschiniKarakayaliValenzuela}, Foschini et al.  compared the performance of conventional cellular networks (CCN)  to that of  coordinated transmission, and showed an enormous gain in using coordinated transmission in terms of spectral efficiency. A detailed overview of cooperative cellular networks and several possible degrees of cooperation can be found in \cite{GesbertHanlyShamaiYu}. In \cite{DingVincent}, authors presented the performance of distributed antenna arrays in cloud radio access networks with spatially random BSs and the minimal number of BSs to meet a predefined quality of service is analyzed using stochastic geometry.  However,  the interference from BSs outside the cluster which is an performance limiting  factor is not accounted.   In \cite{Gesbert:2008icc}, CoMP is analyzed for sum rate maximization in uplink using linear beam forming with the assumption that BSs have full local and non-local channel state information.  

 Traditional grid model based approach is used in \cite{heath}, to obtain the fundamental limits on the capacity of a cloud radio network. The authors have suggested that even with a faster backhaul or more efficient signal processing, the gain in capacity from a cloud radio network as opposed to a conventional cellular network cannot be improved due to inter-cluster interference. 

The size of the cloud is limited by the propagation delay in optical fiber communication and other implementation constraints. Hence cloud processing is used among  BSs grouped geographically (clusters). In this case, users located at the boundaries between the clusters receive interference from the BSs of neighboring clusters that can not be suppressed. The achievable rate as obtained by zero-forcing  dirty paper coding (ZF-DPC) with complete CSI takes a hit because the inter-cluster interference cannot be suppressed. In \cite{spcomKuchi14},  ZF-DPC method with complete and limited channel feedback is analyzed with distance depended interference. It is shown that even with a limited feedback of  three to six   channel states from each user, the cloud radio has the potential to offer a significant increase in the capacity compared to conventional systems. However,   clustering of the BSs  is not considered. The impact of clustering with distance dependent interference and full-channel state information  is considered in our earlier work \cite{STVKKRK_ICC15}. 

 Cooperative multipoint transmission with partial cooperation is studied in \cite{Mennerich:2011}. This paper considered user centric cooperation and assumes each user report a subset of strongly interfering BSs.  Multiple antennas at each of the cooperating cells  utilize wideband beamforming to steer the signal more towards the  centre of cluster area and showed that,  partial cooperation  promises high gains for most UEs.  In \cite{Garcia:2010}, dynamic clustering is considered, wherein cooperative clusters periodically regroup. It can improve the cell-edge SE, but, additional complex scheduling and backhaul connections are required. In \cite{Tanbourgi},  a user centric BS clustering along with channel dependent   joint transmission is analyzed.   It is shown that for small cooperative clusters, non-coherent joint transmission by small cells provides spectral efficiency gains without significantly increasing cell load.  A joint transmission scheme is analyzed for heterogeneous networks in \cite{NigamHaenggi}, and is shown that the BS cooperation boosts the coverage probability    by $17$\% for a general user and by $24$\% for a worst case user.

 User centric clustering  requires each user to know their strongest interfering BSs and this set of BSs should be the same for jointly served users (for the BSs to perform appropriate beamforming) and which is typically not the case. In our work, instead of user centric or dynamic BS cooperation, we consider fixed geographically clustering  of BSs which does not require complex scheduling between clusters.  The BSs in a cluster are connected to a cloud processor where the base band signals are processed and send to the remote radio head (RRH) for transmission.   We provide analytical expressions for signal-to-noise plus interference ratio ($\sinr$) distribution  with the residual inter cluster interference for zero-forcing dirty paper coding transmission.

\subsection{Main Contributions}

In this paper, we evaluate the  coverage and  rate in a cloud radio network with distance dependent interference. The main contributions are as follows:
\begin{itemize}
\item \textit{Geographical clustering with complete channel knowledge:} In section \ref{sec:clust_full}, we consider fixed area clusters in which the BSs inside a cluster can share complete channel information and hence cancel the  intra-cluster interference at the users associated to the cluster. The $\sinr$ distribution of a typical user with the above BS clustering model is provided.
 
\item \textit{Clustering and ideal cloud:}
 \emph{Ideal clouds} refers to the the hypothetical network where in all the nodes in the network can cooperate. In an ideal cloud network, by appropriate precoding, the interference can be completely  eliminated and hence the system becomes noise limited. This can provide us with  upper bounds  on the  performance of a network.  For finite clustering, we obtain the optimal cluster size required to achieve an $(1-\epsilon),  0<\epsilon<1$ fraction of the performance of an ideal cloud.
 
 \item \textit{Clutering with limited channel knowledge:} The $\sinr$ distribution  of a typical user  when the user can feedback only a limited   channel channel state information to the cloud processor is obtained in Section \ref{sec:clust_pcsi}.  
 \end{itemize}
%
%
 
\subsection{Organization of the paper}
In Section \ref{sec:sysmodel}, the system model used in this paper is discussed. In Section \ref{sec:clust_full}, the performance of geographical clustering is analyzed. In Section \ref{sec:Ideal}, the performance of limited clustering is compared with an ideal cloud. In Section \ref{sec:clust_pcsi}, the performance of clustering with limited channel feedback is discussed and the paper is concluded in  Section \ref{sec:conc}.
\section{System Model}
\label{sec:sysmodel}
In this section, we provide a mathematical model of the system that will be used in the subsequent analysis. We begin with the spatial distribution of the base stations.
\subsection{Network model} 
The locations of  the  BSs are modeled as a spatial Poisson point process (PPP)  \cite{stoyan} $\Phi_b$ of density $\lambda_b$ and the user equipments (UEs) follow another independent PPP $\Phi_u$ of intensity $\lambda_u$ in $\R^2$.

 We assume a nearest BS connectivity model, \ie, a user, $x\in\Phi_u$, connects to a BS $y\in\Phi_b$, if and only if $\|x-y\|<\|x-z\|$, $\forall z\in \Phi_b$. The nearest BS connectivity model results in a Voronoi tessellation of the plane with respect to the BS locations. Hence the service area of a BS is the Voronoi cell associated with it.
%

\subsection{Channel and pathloss model}
A standard path loss model $\|x\|^{-\alpha}$, $\alpha>2$, is assumed. Independent Rayleigh fading with unit mean is assumed between any pair of BS and UE.  Therefore, the received signal amplitude for a stream from a BS $x\in \Phi_b$ to a UE $y\in\Phi_u$ is given by
$h_{xy}r_{xy}^{-\alpha/2}$, where $r_{xy}$ is the distance between UE $x$ and BS $y$.  Also $h_{xy}$ is the corresponding fading term with $h_{xy} \sim \mathcal{CN}(0,1)$. The noise term is assumed to be circularly symmetric complex Gaussian noise (AWGN) with zero mean and variance $\sigma^2$. Assuming a transmit power of $P_t$ for each transmitter, the average received signal-to-noise-power ratio for a user  located at $x$  from its associated  BS $y$ is $\snr_{xy}=P_t r_{xy}^{-\alpha}/{\sigma^2}$.

\subsection{Geographical Clustering}
We assume that the BS clusters are selected geographically and in particular  assume  square clusters as shown in Fig.\ref{fig:cluster}. Without loss of generality,  we analyze the performance of the users in  the  cluster that contains the origin. To simplify the analysis (integrals), we approximate the cluster that contains the origin  by a disc\footnote{If the approximating disc forms an incircle to the square cell, the analyzed  coverage will be a lower bound to the actual coverage probability. Similarly,  the analyzed  coverage will be an upper  bound  if  the disc forms a circumcircle. See Figure \ref{fig:CovClust}. }.

 The BSs in  a cluster  are connected to a central unit (cloud) and all the channel information generated by the users  is processed at this central unit    to precode the transmit  signals. We assume that the clustering area is defined  only to coordinate the BSs and not for users. A user will always find the nearest BS regardless of the area it belongs to.  We denote the cluster regions by $C_i, i=0,1,\hdots$. In particular $C_0$ denotes the cluster centered at the origin. 
 We assume all the clusters are identical and since $\Phi_b$ is stationary and hence we focus on the cluster $C_0$.
\begin{figure}[h]
\centering
\includegraphics[scale=0.4]{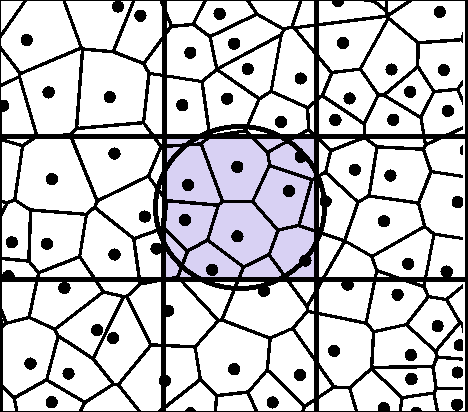}
\caption{Illustration of the a cloud radio system employing clustering. Here the BSs in each rectangular area are coordinated using a central cloud processor. In this paper we assume that the cluster region $C_0$ is a  disc to simplify the analysis.}
\label{fig:cluster}
\end{figure}

\subsection{Received Signal Vector}
Let $m = |C_0\cap\Phi_b|$ denote the size (number of BSs in the cluster) of the cluster $C_0$.
Let ${\X_m = [x_1~x_2~\hdots~x_m]^T}$ be the transmitted symbol vector\footnote{We used $A^T$ to denote transpose of $A$ and $ A^H$ to denote the conjugate transpose.} by the $m$-BSs of the cluster $C_0$. Then the received signal vector of the cluster users denoted by $\Y_m=\left[ y_1, y_2, \hdots, y_m\right]^T$ is given by,

\begin{align}
\begin{bmatrix}
   y_1\\
   y_2\\
   \vdots\\
   y_m
  \end{bmatrix} &= \begin{bmatrix}
  h_{11}r_{11}^{-\alpha/2} & h_{12}r_{12}^{-\alpha/2} & \hdots & h_{1m}r_{1m}^{-\alpha/2}\\
  h_{21}r_{21}^{-\alpha/2} & h_{22}r_{22}^{-\alpha/2} & \hdots & h_{2m}r_{2m}^{-\alpha/2}\\
    \vdots & \vdots & \vdots & \vdots\\
  h_{m1}r_{m1}^{-\alpha/2} & h_{m2}r_{m2}^{-\alpha/2} & \hdots & h_{mm}r_{mm}^{-\alpha/2}
   \end{bmatrix} \begin{bmatrix}   x_1\\
   x_2\\
   \vdots\\
   x_m
 \end{bmatrix} + \begin{bmatrix}   I_1\\
    I_2\\
    \vdots\\
    I_m
   \end{bmatrix} + \begin{bmatrix}   n_1\\
   n_2\\
   \vdots\\
   n_m
  \end{bmatrix},
\end{align} \ie, \begin{align}
\Y_m&=\H_m \X_m+\I_m+\No_m,
\label{eq:clusmain}
\end{align}
where $\No_m = [n_1~n_2~\hdots~n_m]^T$, with $n_i\sim \mathcal{CN}{(0,\sigma^2)}$. Also  $\I_m=\left[I_1, I_2, \hdots, I_m  \right]^T$  are the interfering signals from BSs outside  the cluster $C_0$, where  $I_i =\sum_{b\in\Phi_b^\prime} h_{ib} r_{ib}^{-\alpha/2}x_{b}$ and ~{$\Phi_b^\prime=\Phi_b\setminus\ C_0\cap\Phi_b$} and $\H_m$ is the channel matrix.  In
our model, users associate with the nearest BS, a BS schedules a
randomly selected user associated with it. 
\section{Clustering with complete cooperation}\label{sec:clust_full}
In \cite{dpc}, Costa proved that by dirty paper (DP) precoding, in a network where the interference is non-causally known at the transmitter, it is possible to achieve the
same capacity as if there were no interference. 
A reduced-complexity precoder  is presented in \cite{CaireShamai}. This technique  uses the  $\Lb\Q$ factorisation of the channel matrix to obtain a lower triangular channel matrix which decouples users in a layered manner, and helps in DPC implementation. 
This technique nulls the interference between data streams and hence the name zero forcing dirty paper coding ~(ZF-DPC). We use  this reduced complexity ~ZF-DPC based algorithm to obtain an achievable rate in a  cloud radio network. 
\subsection{Zero-forcing dirty paper coding}
Since the user channels and the transmitted  signals are known non-causally,  the intra-cluster interference   can be canceled by using ZF-DPC using appropriate precoding. Let $\H_m = \Lb_m \Q_m$ be the $\Lb \Q$ decomposition of $\H_m$. Using  $\W_m=\Q_m^H$ as the precoding matrix, the received signal by the users is 
\begin{align}
\Y_m&=\H_m\W_m \X_m+\I_m+\No_m, \nonumber\\
 &= {\Lb_m  \X_m + \I_m+\No_m},\nonumber\\
 \begin{bmatrix}
   y_1\\
   y_2\\
   \vdots\\
   y_m
  \end{bmatrix} &= \begin{bmatrix}
   l_{11}&0&\hdots&0\\
   l_{21}&l_{22}&\hdots&0\\
   \vdots&\vdots&\vdots&\vdots\\
   l_{m 1}&l_{m 2}&\hdots&l_{m m}
 \end{bmatrix} \begin{bmatrix}   x_1\\
   x_2\\
   \vdots\\
   x_m
 \end{bmatrix} + \begin{bmatrix}   I_1\\
    I_2\\
    \vdots\\
    I_m
   \end{bmatrix} + \begin{bmatrix}   n_1\\
   n_2\\
   \vdots\\
   n_m
  \end{bmatrix}  .
\label{rec_sig_HomogClus}
 \end{align}
Hence the  signal received  by a user  $u_i$  is
  \begin{align*}
  {y_i} &= {l_{ii}~x_{i}+\sum_{j<i}l_{ij}x_j +I_i+ n_i} \mbox{\quad for\quad} i=1, 2,\hdots m.
  \end{align*}
  For a user $u_i$ (associated with BS $b_i \in C_0\cap \Phi_b$), the precoder   $\W_m$ helps to  cancel interference from BSs with indices $j>i$ of the cluster.
 The  residual  interference corresponding to BSs $j<i$ can be eliminated by using  DPC successively. A suboptimal implementation of DPC, Tomlinson-Harashima Precoding (THP), can be found in \cite{Windpassinger_THP}.   
 It should also be noted that the signals from other clusters contributing interference to  the users of cluster $C_0$   are also precoded using their own  cluster channel state information. However,  for the users in  $C_0$, these interfering signals also have the same power as the precoding is an  unitary matrix  and will not change the  received  interference power. Therefore, the post processing $\sinr$ for the UE $u_i$ is given by,
    \begin{align}
    \sinr_{i}=\frac{ |l_{ii}|^2}{\sigma^2+\I_i(\Phi)},
    \end{align}
      where $\I_i(\Phi)=\sum_{b\in\Phi_b^\prime}  |h_{ij}|^2 r_{ij}^{-\alpha}$. The downlink rate for the  UE $u_i$ is
  \begin{equation}
  {C_i} = {\log\left( 1+ \frac{|l_{ii}|^2 }{\sigma^2+\I_i(\Phi)} \right)}~\text{bps/Hz.} \label{eq:cidpc_homog_cloud}
  \end{equation}
The diagonal elements of the lower triangular matrix, $\Lb$, are weighted linear combinations of exponential random variables each scaled with distance dependent pathloss, \cite{LCWangCovPerfAnalMUMIMOBCS,CaireOnAchvblthrptMulantGauBC}.
Because of nearest BS connectivity, the off-diagonal terms are small compared to the diagonal terms. Hence, in the matrix $\H_m$, the dominant element in any row is the diagonal element, $h_{ii} r_{ii}^{-\alpha/2}$, where $r_{ii}$ is the distance to the tagged BS from $i$th UE.  Therefore, we can approximate $l_{ii}\approx h_{ii} r_{ii}^{-\alpha/2}$. This key approximation makes the analysis tractable and is validated through simulations.

\subsection{Coverage Probability}
Without lose of generality, we consider a cluster around the origin, \mbox{$D_c=\mathcal{B}(O,R)$} and a typical user, $u$, dropped in this cluster randomly, as shown in Fig.~\ref{clust_anal_setup}. Let $r_u$ denote the distance of this typical user from the origin and $r_b$ denote the distance of the typical user to its closest BS. A user is   in coverage if the received $\sinr$ is above a threshold $T$. We now evaluate the coverage probability of a typical user. 

%


\begin{table}[ht]
\renewcommand{\arraystretch}{1.3}
\centering \caption{Distance distributions}
\begin{tabular}{c c c l l}
\hline
\hline
1.&$r_b$& Nearest distance& ${  f_{r_b}(r_b)=2\pi\lambda r_b e^{-\lambda \pi r_b^2} }$\hspace{1cm}& $0\leq r_b$ \\[5pt]
2.&$r_u$& Distance of user& $f_{r_u}(r_u)=\dfrac{2 r_u}{R^2}$&$0\leq r_u\leq R$\\
\hline
\hline
\end{tabular}
\label{tab:dist_distri}
\end{table}
\begin{figure}
\begin{center}
\begin{tikzpicture}[scale=0.6]
\draw[fill=gray!10,draw=black!0] (0,0)node(c){}  circle (5cm);
\draw[fill=gray!0,draw=black] (2.2,1.8) node(u1) {}  circle (1.5cm);
\draw[draw=black]  (0,0) circle (5cm);
\draw[] (-6.5,0) -- (6.5,0) coordinate (x axis);
\draw[] (0,-6.5) -- (0,6.5) coordinate (y axis);
\coordinate (x) at (0,0);
\coordinate (y) at (2.2,1.8);
\coordinate (y1) at (3.7,1.8);
\coordinate (z1) at (-4.5,-3);
\coordinate (b1) at (3.7,1.8);
\coordinate (b2) at (-0.28,-4.8);
\coordinate (R) at (0,5);
\draw (x) -- (y);
\draw(-1.2,2) node {$\mathcal{B}(0,R)$};
\fill (u1) circle [radius=4pt];
\draw (2.2,1.8) -- (b1);
\path let \p1 = (u1) in node  at (\x1,2.2) {$u$};
\draw (3,1.4) node  {$r_{b}$};
\draw (1,1.4) node  {$r_{u}$};
\node [draw,scale=0.3,regular polygon,regular polygon sides=3,fill=black!100]at (b1) {}; 
\node [draw,scale=0.3,regular polygon,regular polygon sides=3,fill=black!100]at (5.2,5.4) {}; 
\draw (6.2,5.4) node  {BS};
\fill (5.2,4.4) circle [radius=4pt];
\draw (6.2,4.4) node  {UE};
\node [fill,draw,scale=0.3,circle]at (x) {};  
\draw(-0.5,0.4) node {$O$};
\node [draw,scale=0.3,regular polygon,regular polygon sides=3,fill=black!100]at (-2,-3) {}; 
\node [draw,scale=0.3,regular polygon,regular polygon sides=3,fill=black!100]at (-0.4,-0.5) {}; 
\node [draw,scale=0.3,regular polygon,regular polygon sides=3,fill=black!100]at (-3,2) {}; 
\node [draw,scale=0.3,regular polygon,regular polygon sides=3,fill=black!100]at (2,-2) {}; 
\node [draw,scale=0.3,regular polygon,regular polygon sides=3,fill=black!100]at (3,5) {}; 
\node [draw,scale=0.3,regular polygon,regular polygon sides=3,fill=black!100]at (6,2) {}; 
\node [draw,scale=0.3,regular polygon,regular polygon sides=3,fill=black!100]at (1,4) {}; 
\node [draw,scale=0.3,regular polygon,regular polygon sides=3,fill=black!100]at (3,-5) {}; 
\node [draw,scale=0.3,regular polygon,regular polygon sides=3,fill=black!100]at (5.5,-1) {}; 
\node [draw,scale=0.3,regular polygon,regular polygon sides=3,fill=black!100]at (-2,5) {}; 
\node [draw,scale=0.3,regular polygon,regular polygon sides=3,fill=black!100]at (-5.5,1) {}; 
\node [draw,scale=0.3,regular polygon,regular polygon sides=3,fill=black!100]at (-4.5,-1) {}; 
\end{tikzpicture}
\end{center}
 \caption[]{Illustration of a cluster and a typical user. The cluster center is at origin, $O$ and the typical user is at  distance $r_u$ from the origin.   The closest  BSs to this typical user is at a distance  $r_b$.}
 \label{clust_anal_setup}
\end{figure}
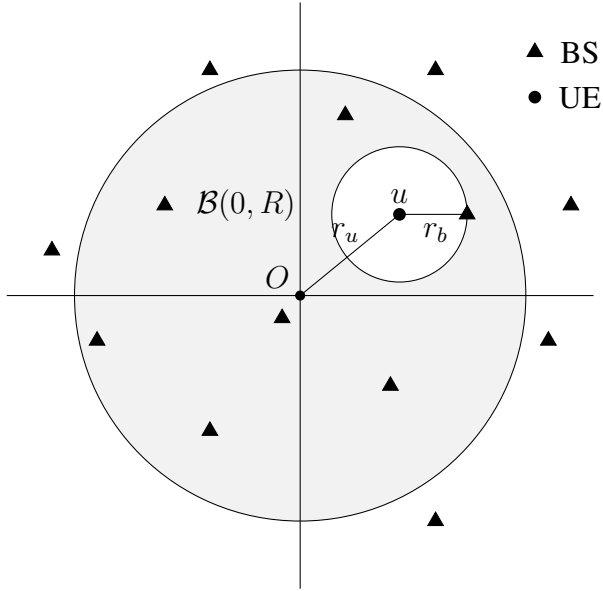
Let
$ \Lambda(\rho,T)=\,_2F_1\left(1,\frac{\alpha-2}{\alpha };2-\frac{2}{\alpha
};-\rho^\alpha T \right),$
where $\,_2F_1(a,b,c,x)$ is the standard hypergeometric function.
\begin{lemma}
The coverage probability of a typical user is lower bounded by
\begin{align}
P_{c, R}(T)\geq & \E_{r_u,r_b}[P_1(T)\i(r_b<R-r_u)]\\\nonumber
&+\E_{r_u,r_b}[P_2(T)\i(R-r_u<r_b<R+r_u)],
\end{align}
where $P_1(T)$ and $P_2(T)$ are given by
\begin{align}
P_1(T) &= e^{-T r_b^\alpha\sigma^2}e^{-\lambda \int_{0}^{2\pi}\frac{T r_b^\alpha(R^\prime)^{2-\alpha } \, \Lambda(r_b/R^\prime,T)}{\alpha -2} \d \varphi},
\label{eq:Full_Clut_Cov_T1}
\end{align}
and
\begin{align}
 P_2(T)&=  e^{-T r_b^\alpha\sigma^2} e^{-\lambda 2 \Theta  \frac{ T r_b^2 \, \Lambda(1,T)}{\alpha -2}} e^{-\lambda\int_{-\Theta}^{\Theta}\frac{T r_b^\alpha \left({R^\prime}\right){}^{2-\alpha } \,
        \Lambda(r_b/R^\prime,T)}{\alpha -2}\d \varphi}.
\label{eq:Full_Clut_Cov_T2}
 \end{align}
Here $R^\prime=\sqrt{\cos^2(\varphi) r_u^2+R^2}-\cos(\varphi) r_u$   and  $\Theta=\pi-\arccos\left[ \frac{r_b^2+r_u^2-R^2}{2 r_b r_u}\right]$. The PDF of $r_b$ and $r_u$ are given in Table \ref{tab:dist_distri}.
 \label{the:cov}
\end{lemma}
\begin{proof}
See Appendix \ref{sec:app_clust_full}.
\end{proof}
The inequality in the above theorem results by neglecting the cases when there are no BSs inside the cluster $C_0$ and the probability of this event decreases to zero with increasing BS density or cluster radius. 
The integrals in \eqref{eq:Full_Clut_Cov_T1} and \eqref{eq:Full_Clut_Cov_T2} can be evaluated using a M-quadrant approximation by dividing the interfering area into $2\pi/M$ angular regions \cite{STVKKRK_ICC15} as follows: 
\begin{align*}
\int_{0}^{2\pi}\frac{T r_b^\alpha(R^\prime)^{2-\alpha } \, \Lambda(r_b/R^\prime,T)}{\alpha -2} \d \varphi&\approx\sum_{n=0}^{M-1}\frac{2\pi}{M}\frac{T r_b^\alpha {R^\prime_n}^{2-\alpha } \, \Lambda(r_b/R_n^\prime,T)}{\alpha -2},
\end{align*} where, $R_n^\prime=\sqrt{\cos ^2(\theta_n ) r_u^2+R^2}-\cos (\theta_n ) r_u$. This approximation can be used for numerical evaluation. In this paper, we use $M=128$ for numerical evaluation. 
We now numerically evaluate the $\snr$ distribution given in Lemma \ref{the:cov} and compare with Monte-Carlo simulations. We choose $\alpha =4$, and  the noise variance $\sigma^2$ is chosen so that the average $\snr$ at a receiver at a distance $200$ m is $10$ dB. The figures are  also parameterized by the average inter BS distance $D$ and equals  $D=2/\sqrt{\pi\lambda}$, where $\lambda$ is the BS density.  
\begin{figure}[ht]
\centering
\begin{tikzpicture}
\begin{semilogyaxis}[scale=1.35,
grid = both,
legend style={
 cells={anchor=west},
legend pos=south west ,
 },
xlabel = $\sinr$ threshold $T$,
ylabel = Coverage probability
]
\addplot
coordinates{
(-5.000000, 0.774519) (-2.000000, 0.650059) (1.000000, 0.511849) (4.000000, 0.383225) (7.000000, 0.278284) (10.000000, 0.199064) (13.000000, 0.141486) (16.000000, 0.100311) (19.000000, 0.071053)
};\addlegendentry{No cloud}
\addplot
coordinates{
(-5.000000, 0.916143) (-2.000000, 0.860017) (1.000000, 0.784771) (4.000000, 0.695227) (7.000000, 0.598151) (10.000000, 0.499965) (13.000000, 0.406278) (16.000000, 0.321577) (19.000000, 0.248691)
 };
\addlegendentry{$(\frac{1}{2},200)$}
\addplot
coordinates{
(-5.000000, 0.818357) (-2.000000, 0.721495) (1.000000, 0.610874) (4.000000, 0.499807) (7.000000, 0.398132) (10.000000, 0.310441) (13.000000, 0.237774) (16.000000, 0.179361) (19.000000, 0.133569)
 };\addlegendentry{$(\frac{1}{2},400)$}
\addplot
coordinates{
(-5.000000, 0.941403) (-2.000000, 0.900167) (1.000000, 0.842180) (4.000000, 0.768705) (7.000000, 0.682858) (10.000000, 0.588852) (13.000000, 0.492173) (16.000000, 0.398962) (19.000000, 0.314507)
 };
\addlegendentry{$(1,200)$}
\addplot
coordinates{
(-5.000000, 0.851752) (-2.000000, 0.767501) (1.000000, 0.666337) (4.000000, 0.558962) (7.000000, 0.455159) (10.000000, 0.361243) (13.000000, 0.280388) (16.000000, 0.213522) (19.000000, 0.160045)
 };\addlegendentry{$(1,400)$}
\addplot
coordinates{
(-5.000000, 0.981727) (-2.000000, 0.967850) (1.000000, 0.946569) (4.000000, 0.915749) (7.000000, 0.872377) (10.000000, 0.812924) (13.000000, 0.735608) (16.000000, 0.642859) (19.000000, 0.541520)
 };\addlegendentry{$(10,200)$}
 \addplot
 coordinates{
 (-5.000000, 0.993711) (-2.000000, 0.987676) (1.000000, 0.976229) (4.000000, 0.955386) (7.000000, 0.919783) (10.000000, 0.864126) (13.000000, 0.786060) (16.000000, 0.688679) (19.000000, 0.580302)
  };\addlegendentry{Ideal}

 \addplot[
    mark =oplus,  mark size=3, only marks]
coordinates{
(-5.000000, 0.961682) (-2.000000, 0.927102) (1.000000, 0.874274) (4.000000, 0.802475) (7.000000, 0.710449) (10.000000, 0.606615) (13.000000, 0.493923) (16.000000, 0.386067) (19.000000, 0.292283)
 };
\addlegendentry{MC-Sim $(1,200)$}
\addplot[
   mark=triangle*, mark options={fill=white},  mark size=3, only marks]
coordinates{
(-5.000000, 0.874826) (-2.000000, 0.827668) (1.000000, 0.760821) (4.000000, 0.676865) (7.000000, 0.580081) (10.000000, 0.477657) (13.000000, 0.379215) (16.000000, 0.292804) (19.000000, 0.219046)
 };\addlegendentry{MC-Rect$(\frac{1}{2},200)$}
 \end{semilogyaxis}
\end{tikzpicture}
\caption[]{Coverage probability for various combinations of cluster area and inter base station distance $(A, D)$ where $A=\pi R^2$ is in km${}^2$ and $D$ in m. The conventional cellular networks with no cooperation  is denoted by   "\textit{No cloud}"  and   "\textit{Ideal}" is the scenario in which  the entire network is cooperating, \ie, $R=\infty$. The Monte-Carlo simulation for $(A,D) = (1,200)$ is marked  by $\oplus$. MC-Rect  corresponds to a Monte-Carlo simulation with a  square  cell  $C_0$  of   area  $A$  instead of the circular approximation.  }
\label{fig:CovClust}
\end{figure}
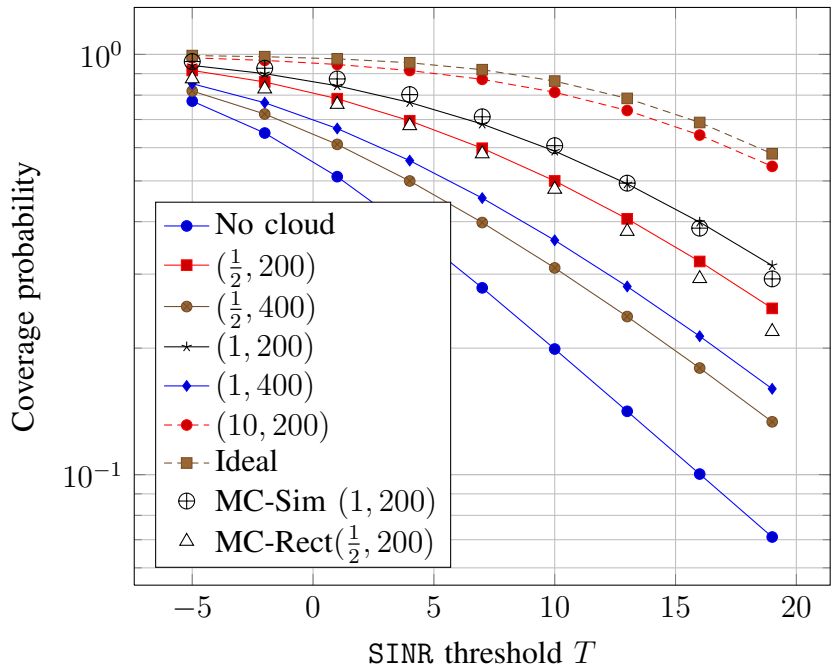

In Figure \ref{fig:CovClust}, the coverage probability given in Lemma \ref{the:cov} is plotted as a function of the $\sinr$ threshold $T$ for various cluster sizes  $A= \pi R^2$ and the inter-BS distances $D$. Also the exact Monte-Carlo simulations are plotted for some configurations. The Monte-Carlo simulations utilize  a rectangular tessellation for BS cooperation. We observe that the coverage probability obtained in Lemma \ref{the:cov} matches the Monte-Carlo simulation closely. This observation also validates our assumption $l_{ii}\approx h_{ii} r_{ii}^{-\alpha/2}$.
  As expected, the coverage increases with increasing cluster radius when the BS density (or $D$) is constant. Also, for a given cluster radius, the coverage probability increases with decreasing $D$ (or increasing density). This is because, with cloud clustering, the interferers inside a cluster will be completely eliminated and most of the users in the cluster are associated to  a closer BS inside the cluster and hence have a higher signal strength. This is in contrast to the performance of a conventional cellular network where the coverage does not change with density \cite{ganti_coverage}.

\begin{table} [ h b t]
\setlength{\tabcolsep}{4pt}
\renewcommand{\arraystretch}{1}
\centering
\begin{subtable}{.4\textwidth}
\centering
\begin{tabular}{l c c c c c}
  \hline
$A$&$\bar{N}$& 5 \% & 10 \% & 50 \%& Mean\\
\hline
Cell&-&0.09&0.15&1.23&2.14\\
1/2&4& 0.21 & 0.40 & 2.37 &  3.41 \\
1&8& 0.29 & 0.54 & 3.00 &  3.86 \\
2&16& 0.41 & 0.72 & 3.65 &  4.43 \\
10&80& 0.80 & 1.49 & 5.69 &  6.17 \\
Ideal&$\infty$& 2.48 & 3.54 & 7.74 & 8.22 \\
\hline
\end{tabular}
\label{tab:rcdf_lim_400}
\caption{$D=400$}
\end{subtable}
\begin{subtable}{.4\textwidth}
\centering
\begin{tabular}{l c c c c c}
  \hline
$A$&$\bar{N}$& 5 \% & 10 \% & 50 \%& Mean\\
\hline
Cell&-&0.09&0.15&1.23&2.14\\
1/8&4& 0.20 & 0.39 & 2.44 &  3.43 \\
1/4&8& 0.28 & 0.51 & 2.98 &  3.87 \\
1/2&16& 0.37 & 0.70 & 3.64 &  4.41 \\
2.5&80& 0.78 & 1.46 & 5.57 &  6.05 \\
Ideal&$\infty$& 2.48 & 3.54 & 7.74 & 8.22 \\
\hline
\end{tabular}
\label{tab:rcdf_lim_200_2}
\caption{$D=200$}
\end{subtable}
\hfill
\caption{Rate profile of a typical user for various combinations of cluster area and inter BS distance, where $A=\pi R^2$ is in km${}^2$ and $D$ in m. $N_{avg}$ is the average number of BSs in a cluster. Rates are in bits/sec/Hz.}
\label{tab:rcdf_clust}
\end{table}

The cumulative distribution function (CDF) of rate, $C=\log_2(1+\sinr)$, for a user can be obtained from the coverage probability as $F_C(t)=\mathbb{P}[\log_2(1+\sinr)<t]=1-P_{cR}(2^t-1)$ where $t$ is the rate in bits/sec/Hz.
The rate profile is provided in  Table \ref{tab:rcdf_clust}, for various $R$ and ~$D=200,\, 400$.   Compared to the conventional cellular system,  the $5\%$  percentile point (corresponds to edge users)  increases by  $300\%$, while the mean rate increases  by $200\%$  when  $16$ BSs  cooperate in a cluster of area $ 1/2$ km$^2$. We also observe that the mean rate increases faster to the ideal case of complete cooperation than the $5\%$ rate, indicating the disparity between the edge and the interior users in a cluster.

\section{Ideal clouds and Minimum Clustering Area with $\epsilon$-penalty}\label{sec:Ideal}
Ideal cloud is an hypothetical network in which all the nodes cooperate and eliminate co-channel interference.  It is easy to observe that   the performance of an ideal cloud  provides upper bounds on the achievable performance of a network with limited cooperation. 
In this section, we look at how the cluster radius  $R$ should be chosen so that the coverage probability (a.k.a. $\sinr$ distribution) is close to the coverage probability of an ideal cloud.
\subsection{Coverage probability of Ideal clouds}
Since there is no interference, the coverage probability equals the probability that the  received signal-power-to-noise ratio is larger than a threshold  and is given by \cite{spcomKuchi14},
\begin{align}
  {P}_{c,\infty}(T)=2 \pi  \lambda_b \int_{0}^\infty r e^{-T\sigma ^2
r^{\alpha}-\pi \lambda_b  r^2}\mathbf{d}r.
\end{align}
This can be obtained from Lemma \ref{the:cov}, by noting that when $R=\infty$, there will be only type-I users (see Appendix \ref{sec:app_clust_full}). 
\subsection{ Minimum clustering area with $\epsilon$-penalty}
 When all the BSs are cooperating, the network becomes noise limited instead of being interference limited. In this case, the throughput  can be increased   by merely  increasing the transmit power or increasing the density of the BS as shown in Figure \ref{fig:dpc_tic}.
However, due to practical limitations, only the BSs in  a small area can cooperate  using the  cloud infrastructure. In particular, all the BS in the cluster $C_0 =B(O,R)$ cooperate and precoding removes  the intra-cluster interference.

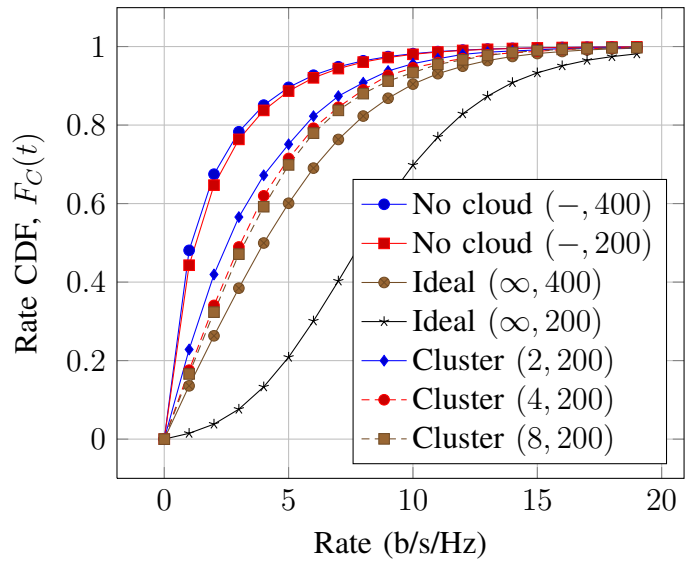
\begin{figure}[h]
\centering
\begin{tikzpicture}
\begin{axis}[scale=1.1,
grid = both,
legend style={
 cells={anchor=west},
legend pos=south east ,
 },
xlabel ={Rate (b/s/Hz)},
ylabel = {Rate CDF, $F_C(t)$}
]
 \addplot
 coordinates{
 (0.000000, -0.000000) (1.000000, 0.480885) (2.000000, 0.674924) (3.000000, 0.782970) (4.000000, 0.850860) (5.000000, 0.896055) (6.000000, 0.927037) (7.000000, 0.948599) (8.000000, 0.963723) (9.000000, 0.974372) (10.000000, 0.981887) (11.000000, 0.987195) (12.000000, 0.990947) (13.000000, 0.993599) (14.000000, 0.995474) (15.000000, 0.996800) (16.000000, 0.997737) (17.000000, 0.998400) (18.000000, 0.998868) (19.000000, 0.999200)
  };\addlegendentry{No cloud $(-,400)$ }
\addplot
coordinates{
(0.000000, 0.000000) (1.000000, 0.443025) (2.000000, 0.646995) (3.000000, 0.763569) (4.000000, 0.837368) (5.000000, 0.886615) (6.000000, 0.920402) (7.000000, 0.943922) (8.000000, 0.960421) (9.000000, 0.972040) (10.000000, 0.980239) (11.000000, 0.986030) (12.000000, 0.990123) (13.000000, 0.993016) (14.000000, 0.995062) (15.000000, 0.996508) (16.000000, 0.997531) (17.000000, 0.998254) (18.000000, 0.998766) (19.000000, 0.999127)
 };\addlegendentry{No cloud $(-,200)$}

 \addplot
 coordinates{
 (0.000000, 0.000000) (1.000000, 0.135718) (2.000000, 0.263166) (3.000000, 0.384522) (4.000000, 0.499580) (5.000000, 0.600882) (6.000000, 0.690359) (7.000000, 0.763338) (8.000000, 0.822934) (9.000000, 0.868334) (10.000000, 0.904307) (11.000000, 0.930935) (12.000000, 0.949900) (13.000000, 0.964358) (14.000000, 0.974750) (15.000000, 0.982234) (16.000000, 0.987403) (17.000000, 0.991386) (18.000000, 0.994073) (19.000000, 0.995851)
  };\addlegendentry{Ideal $(\infty,400)$}
  \addplot
  coordinates{
  (0.000000, 0.000000) (1.000000, 0.014755) (2.000000, 0.038209) (3.000000, 0.076256) (4.000000, 0.133119) (5.000000, 0.208602) (6.000000, 0.300966) (7.000000, 0.402738) (8.000000, 0.508892) (9.000000, 0.611405) (10.000000, 0.698164) (11.000000, 0.769523) (12.000000, 0.828769) (13.000000, 0.873744) (14.000000, 0.908344) (15.000000, 0.933086) (16.000000, 0.951063) (17.000000, 0.965013) (18.000000, 0.974388) (19.000000, 0.981443)
   };\addlegendentry{Ideal $(\infty,200)$}
\addplot
  coordinates{
  (0.000000, 0.000000) (1.000000, 0.228376) (2.000000, 0.419415) (3.000000, 0.565650) (4.000000, 0.672060) (5.000000, 0.751089) (6.000000, 0.822651) (7.000000, 0.873678) (8.000000, 0.907903) (9.000000, 0.937772) (10.000000, 0.956441) (11.000000, 0.968886) (12.000000, 0.980087) (13.000000, 0.985688) (14.000000, 0.988799) (15.000000, 0.991288) (16.000000, 0.995022) (17.000000, 0.995644) (18.000000, 0.996889) (19.000000, 0.996889)
   };\addlegendentry{Cluster $(2,200)$}
 \addplot
 coordinates{
 (0.000000, 0.000000) (1.000000, 0.175663) (2.000000, 0.340653) (3.000000, 0.489478) (4.000000, 0.619701) (5.000000, 0.714852) (6.000000, 0.792010) (7.000000, 0.845380) (8.000000, 0.889600) (9.000000, 0.926502) (10.000000, 0.945715) (11.000000, 0.959134) (12.000000, 0.969198) (13.000000, 0.979262) (14.000000, 0.983837) (15.000000, 0.987801) (16.000000, 0.992376) (17.000000, 0.994815) (18.000000, 0.996035) (19.000000, 0.998170)
  };\addlegendentry{Cluster $(4,200)$}

   \addplot
       coordinates{
       (0.000000, 0.000000) (1.000000, 0.165985) (2.000000, 0.323474) (3.000000, 0.471208) (4.000000, 0.592039) (5.000000, 0.698553) (6.000000, 0.779421) (7.000000, 0.837476) (8.000000, 0.880113) (9.000000, 0.912052) (10.000000, 0.933763) (11.000000, 0.954688) (12.000000, 0.966174) (13.000000, 0.977502) (14.000000, 0.984896) (15.000000, 0.989616) (16.000000, 0.992763) (17.000000, 0.994808) (18.000000, 0.996539) (19.000000, 0.997011)
        };\addlegendentry{Cluster $(8,200)$}
 \end{axis}
\end{tikzpicture}
\caption[Rate distribution of cloud network with average inter BS distance $400$m and $200$m.]{Rate distribution of cloud network with average inter BS distance $D=400$m and $200$m.}
		\label{fig:dpc_tic}
\end{figure}
Let $y=(R(1-\delta),0)$ denote the location of the user at a distance $(1-\delta)R$ from the center of the cluster. A small $\delta$ implies that the user is close to the cluster boundary and $\delta$ close to one implies an interior user. 
We now will evaluate the cluster size $R$ so that that the coverage probability with finite cluster radius $R$ equals $(1-\epsilon)P_{c, \infty}(T)$, \ie,
\[R^*(\epsilon, \delta) = \argmin_R\{P_{c,R}(T) =(1-\epsilon)P_{c, \infty}(T)\},\]
where
\[P_{c,R}(T)=\P\left(\frac{ |h|^2 r^{-\alpha}}{\sigma^2+ \I_i(\Phi)}\geq T\right) .\]
\begin{lemma}\label{lem:scaling}
The optimal scaling of the BS cluster radius is 
\begin{equation}
R^*(\epsilon, \delta) \sim\left( \lambda  T \Delta(\delta)\eta(T) \epsilon^{-1}\right)^{1/(\alpha-2)}, \quad \epsilon \to 0,
\end{equation}
where $\eta(T) =\E[\exp(-T\sigma^2 r^\alpha)r^\alpha]/\E[\exp(-T\sigma^2 r^\alpha) ]$ and  $\Delta(\delta) =\int_{B(o,1)^c}  \frac{1}{\|z-(1-\delta, 0)\|^{\alpha}}\d z$.
\label{lem:345}
\end{lemma}
\begin{proof}
	See Appendix \ref{sec:app_scaling}.
\end{proof}

From Lemma \ref{lem:345}, we observe that  for the coverage to be $1-\epsilon$ close to the ideal cloud, $R^*(\epsilon, \delta)$ should scale as $\epsilon^{-1/(\alpha-2)}$ and tends to infinity when $\alpha \to 2$.  It is easy to see that $\Delta(\delta) \to \infty$ as $\delta \to 1$ and hence $R^*(\epsilon, \delta)$ increases to infinity. This is intuitive since the coverage probability of a user at the boundary of a cooperating cluster always sees some interference. 

When $\sigma^2 =0$, \ie, the noise is ignored,  $\eta(T)= \pi ^{-\alpha /2} \lambda ^{-\alpha /2} \Gamma \left(\frac{\alpha }{2}+1\right)$.  Hence from Lemma \ref{lem:345}, we have that the optimal scaling should be 
\[R^*(\epsilon, \delta) \sim\left(T \Delta(\delta)\pi ^{-\alpha /2} \lambda ^{(2-\alpha) /2} \Gamma \left(\frac{\alpha }{2}+1\right) \epsilon^{-1}\right)^{1/(\alpha-2)},\]
which implies that the average number of nodes in a cooperating cluster should scale as
\[\E[\Phi( B(o,R^*(\epsilon, \delta)))] =\lambda \pi R^*(\epsilon, \delta)^2\sim \left(T \Delta(\delta)\pi ^{-1}  \Gamma \left(\frac{\alpha }{2}+1\right) \epsilon^{-1}\right)^{2/(\alpha-2)}. \]
The above equation shows that the performance is determined by the average number of nodes in a cluster rather than the cluster radius.

\section{Clustering with limited Channel knowledge}\label{sec:clust_pcsi}
As the size of the cluster, $m = |C_0\cap\Phi_b|$,  increases, the amount of CSI required for precoding increases as ${O(m^2)}$ as the  size of $\H_m$ is $m\times m$. For instance, in a cluster with ${m=2}$, the number of individual channel states required is 4, whereas this number increases to 400 for a cloud with 20 BSs. Obtaining such large amount of  CSI in a cellular network is practically infeasible. Therefore, there arises a need to explore the effect of partial CSI on the achievable rate using ZF-DPC. In current cellular networks, the  channel states of at most six best BSs can be estimated  \cite{garg199995}. Next, we discuss  the performance of ZF-DPC  with limited channel feedback. 
\subsection{ZF-DPC with partial CSI}
 We assume that  every user can report the  channel states of its nearest $L\leq m$ BSs of the cluster.
Let the channel matrix with partial CSI from the $m$ users  be denoted by $\H_p \in \R^{m\times m}$  with  the unknown channels  assumed to be zero. Essentially, each row of $\H_p$ has at most $L$ non zero entries.  Let $\H_p=\Lb_p\Q_p$ be the LQ factorization of the channel matrix of $\H_p$.  Hence,  the actual channel matrix denoted by $\H_m$ equals $\H_m=\H_p+\H_r$, where $\H_r \in \R^{m\times m}$ is the matrix consisting of all the unknown channels (not fed back by the users to the central processing node). Using ${\W_m =\Q_p^\dagger}$ as the precoding matrix and the channel being $\H_m$, we obtain,
\begin{align}
\Y_m&=\left(\H_p+\H_r\right)\Q_p^\dagger \X_m +\I_i(\phi)+\No_m\nonumber\\ 
&=\Lb_p \X_m + \underbrace{\H_r\Q_p^\dagger\X_m}_{\Eb_m}+\I_i(\phi)+\No_m.
\end{align}
By using DPC we can cancel interference from known channels, \ie,  the terms in the lower part of $\Lb_p$.  Hence the interference of the BS whose channels are known can be canceled.  However, the interference from the intra cluster BSs whose CSI is unknown leads to the interference term $\Eb_m$, thus degrading the overall performance. 
\subsection{Coverage probability}
Consider a cluster around the origin, $D_c=\B(0,R)$ and a user $u$ at a distance $r_u$ from origin which is tagged to a BS at a distance $r_b$ from it. After  the  interference from nearest $L$ BSs, which is at a distance $r_l$ from the user, is cancelled, the  $\sinr$  is      \begin{align}
    \sinr&=\dfrac{|h_{b}|^2r_{b}^{-\alpha}}{\sigma^2+\sum_{y\in\Phi_b^\prime}|h_y|^2 \|y-u \|^{-\alpha}},
     \end{align}
     where $\Phi_b^\prime=\Phi_b\setminus\{x_1,x_2,\hdots, x_L\}$  and $x_1,x_2,\hdots, x_L\in D_c$ are the $L$ nearest BSs   from the user $u$, and inside the cluster $C_o$. 

\begin{lemma}
The coverage probability of a typical user in a cloud radio with partial channel feedback is 
\begin{align}
P_{c, R}(T)\geq & \E_{r_u,r_b,r_l}[P_1(T)\i(r_b<R-r_u,r_b<r_l<R-r_u)]\\\nonumber
&+\E_{r_u,r_b,r_l}[P_2(T)\i(r_b<R-r_u,R-r_u<r_l<R+r_u)]\\\nonumber
&+\E_{r_u,r_b,r_l}[P_3(T)\i(R-r_u<r_b<R+r_u,r_b\leq r_l \leq R+r_u)]
\end{align}
 where 
\begin{align*}
P_1(T) &= e^{-T r_b^\alpha\sigma^2}e^{-2\pi\lambda\frac{Tr_b^\alpha r_l^{2-\alpha } \, _ \Lambda(r_b/r_l,T)}{\alpha -2}},
\end{align*}
\begin{align*}
 P_2(T)&=  e^{-T r_b^\alpha\sigma^2} e^{-\lambda \int_{-\Theta}^{\Theta}\frac{T r_b^\alpha \left({R^\prime}\right){}^{2-\alpha } \,
      \Lambda(r_b/R^\prime,T)}{\alpha -2}\d \varphi}e^{ -\lambda 2(\pi-\Theta)\frac{T r_b^{\alpha } r_l^{2-\alpha } \,
      \Lambda(r_b/r_l,T)}{\alpha -2} },
 \end{align*}
  and
\begin{align*}
 P_3(T)&=  e^{-T r_b^\alpha\sigma^2}e^{-2\lambda\int_{\Theta_1}^{\Theta_2}\frac{Tr_b^\alpha \left({R^\prime}\right){}^{2-\alpha } \,
       \Lambda(r_b/R^\prime,T)}{\alpha -2}\d \varphi} e^{ -\lambda 2\Theta_1\frac{ T r_b^2 \, \Lambda(1,T)}{\alpha -2}}e^{ -\lambda\int_{\Theta_2}^{2\pi-\Theta_2}\frac{ Tr_b^\alpha r_l^{2-\alpha } \, \Lambda(r_b/r_l,T)}{\alpha -2}\d \varphi}.
\end{align*}
Here $\Theta=\pi-\arccos\left[ \frac{r_l^2+r_u^2-R^2}{2 r_l r_u}\right]$ and $R^\prime=\sqrt{\cos(\varphi)^2 r_u^2-r_u^2+R^2}-\cos(\varphi) r_u$. Also $\Theta_1=\pi-\arccos\left[ \frac{r_b^2+r_u^2-R^2}{2 r_b r_u}\right]$ and $\Theta_2=\pi-\arccos\left[ \frac{r_l^2+r_u^2-R^2}{2 r_l r_u}\right]$. 
The pdf of $r_l$ conditioned on $r_b$ is given by
$\mathbf{f}_{r_l|r_b}(r_l|r_b)= \frac{2(\pi\lambda)^{L}}{(L-1)!}
\left(r_l^2-r_b^2\right)^{L-1}r_l
e^{-\pi\lambda(r_l^2-r_1^2)}$. \label{the:cov_lim}
\end{lemma}
\begin{proof}
See Appendix \ref{sec:app_clust_lim}.
\end{proof}

\subsection{Results and discussion}
We have used the same parameters as in the  previous section for a fair comparison. The path loss is  chosen as $\alpha =4$, and  the noise variance $\sigma^2$ is chosen so that the average $\snr$ at a receiver at a distance $200$ m is $10$ dB and  the average inter BS distance $D$ and equals  $D=2/\sqrt{\pi\lambda}$, where $\lambda$ is the BS density.  In Fig. \ref{fig:CovPCSI}, we have shown coverage probability of a typical user for various combinations of cluster area, available CSI and BS density. We first observe that   the Monte Carlo simulations are very close to the theoretical results.  From Fig.~\ref{fig:cov_pcsi_fixarea}, we can see that coverage increases with available CSI as expected. We can see the significant performance improvement  even with  the channel knowledge of the two nearest BSs, compared to a  conventional cellular network.

We have seen that for clustering with complete channel knowledge, the performance is improved by increasing the number of cooperating nodes, \ie, the dependence of performance is more on number of cooperating nodes. A similar observation can be seen for clustering with partial channel state information. From Fig. \ref{fig:cov_pcsi_fixL}, we can see that, performance of network with $L=4$ is almost similar for cluster area, $A=0.63$ and $1$ for fixed BS density.  For  a fixed area we can see that the coverage improvement with the density (implies increasing average number of BSs per cluster).  This is because of most of the users will be associated with a closer BS.

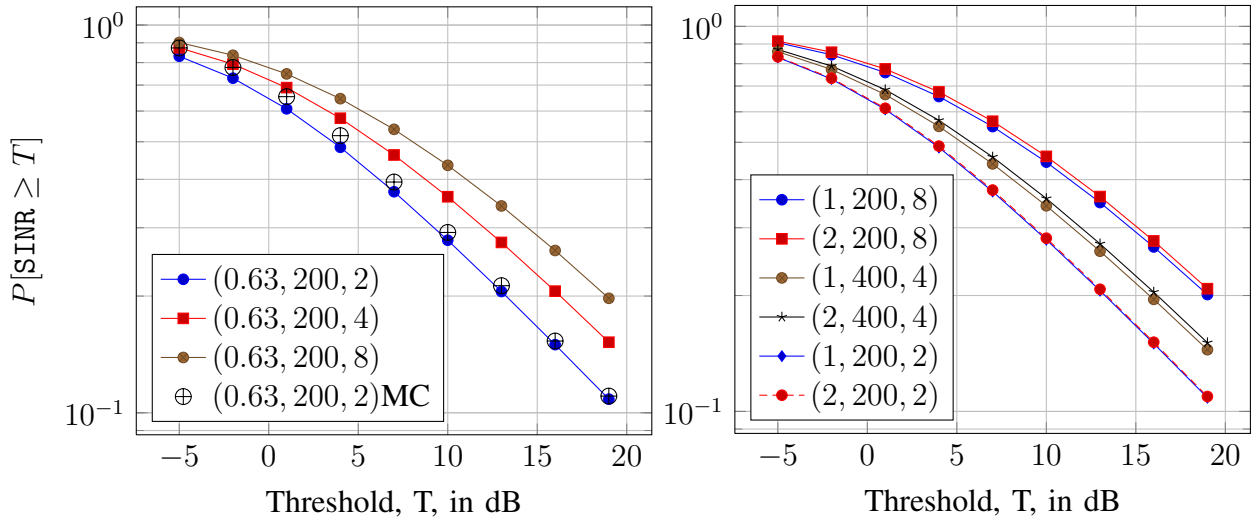
\begin{figure*}%
\centering
\hfill\begin{subfigure}{.5\columnwidth}
\begin{tikzpicture}
	\begin{semilogyaxis}[scale=1,
	grid = both,
legend style={
 cells={anchor=west},
legend pos=south west ,
 },
	xlabel ={Threshold, T, in dB},
	ylabel = {$P[\sinr\geq T]$}
	]
	
 \addplot
coordinates{
	(-5.000000, 0.830730) (-2.000000, 0.729452) (1.000000, 0.608053) (4.000000, 0.483812) (7.000000, 0.371640) (10.000000, 0.278620) (13.000000, 0.205488) (16.000000, 0.149869) (19.000000, 0.108458)
}; \addlegendentry{$(0.63, 200, 2)$}

\addplot
	coordinates{
		(-5.000000, 0.873791) (-2.000000, 0.792939) (1.000000, 0.689578) (4.000000, 0.575006) (7.000000, 0.462259) (10.000000, 0.360746) (13.000000, 0.275020) (16.000000, 0.205909) (19.000000, 0.152051)
};\addlegendentry{$(0.63, 200, 4)$}
\addplot
		coordinates{
			(-5.000000, 0.901547) (-2.000000, 0.835778) (1.000000, 0.748205) (4.000000, 0.645764) (7.000000, 0.538296) (10.000000, 0.434766) (13.000000, 0.341564) (16.000000, 0.262146) (19.000000, 0.197407)
};\addlegendentry{$(0.63, 200, 8)$}

\addplot[
    mark =oplus,  mark size=3, only marks]
coordinates{
	(-5.000000, 0.872893) (-2.000000, 0.777378) (1.000000, 0.653546) (4.000000, 0.518459) (7.000000, 0.393741) (10.000000, 0.291782) (13.000000, 0.212712) (16.000000, 0.153146) (19.000000, 0.110423)
}; \addlegendentry{$(0.63, 200, 2)$MC}
	\end{semilogyaxis}
	\end{tikzpicture}
\caption{Coverage probability for various levels of available CSI with fixed area and density.}%
\label{fig:cov_pcsi_fixarea}%
\end{subfigure}\hfill%
\begin{subfigure}{.45\columnwidth}
\begin{tikzpicture}
	\begin{semilogyaxis}[scale=1,
	grid = both,
legend style={
 cells={anchor=west},
legend pos=south west ,
 },
	xlabel ={Threshold, T, in dB}
	]
	
\addplot
coordinates{
(-5.000000, 0.907090) (-2.000000, 0.844056) (1.000000, 0.758847) (4.000000, 0.657419) (7.000000, 0.549339) (10.000000, 0.444093) (13.000000, 0.348802) (16.000000, 0.267443) (19.000000, 0.201134) 
 };\addlegendentry{$(1, 200, 8)$}

\addplot
coordinates{
(-5.000000, 0.915721) (-2.000000, 0.856893) (1.000000, 0.775394) (4.000000, 0.675793) (7.000000, 0.567211) (10.000000, 0.459748) (13.000000, 0.361500) (16.000000, 0.277200) (19.000000, 0.208360) 
 };\addlegendentry{$(2, 200, 8)$}
\addplot
coordinates{
(-5.000000, 0.859834) (-2.000000, 0.773096) (1.000000, 0.665463) (4.000000, 0.549961) (7.000000, 0.439494) (10.000000, 0.342036) (13.000000, 0.260697) (16.000000, 0.195459) (19.000000, 0.144665)
 };\addlegendentry{$(1, 400, 4)$}

\addplot
coordinates{
(-5.000000, 0.871045) (-2.000000, 0.789055) (1.000000, 0.684900) (4.000000, 0.570217) (7.000000, 0.457996) (10.000000, 0.357337) (13.000000, 0.272497) (16.000000, 0.204140) (19.000000, 0.150856)
 };\addlegendentry{$(2, 400, 4)$}
\addplot
 coordinates{
 (-5.000000, 0.831056) (-2.000000, 0.730083) (1.000000, 0.608971) (4.000000, 0.484820) (7.000000, 0.372514) (10.000000, 0.279247) (13.000000, 0.205870) (16.000000, 0.150063) (19.000000, 0.108528) 
  };\addlegendentry{$(1, 200, 2)$}
\addplot
coordinates{
(-5.000000, 0.833653) (-2.000000, 0.733616) (1.000000, 0.613017) (4.000000, 0.488740) (7.000000, 0.375825) (10.000000, 0.281778) (13.000000, 0.207677) (16.000000, 0.151295) (19.000000, 0.109343) 
 };
\addlegendentry{$(2, 200, 2)$}
\end{semilogyaxis}
	\end{tikzpicture}
\caption{Coverage probability for fixed CSI  and different cluster area and BS density}%
\label{fig:cov_pcsi_fixL}%
\end{subfigure}\hfill
\caption{Coverage probability for various combinations of cluster area, inter base station distance and available CSI  denoted by $(A, D, L)$ where $A=\pi R^2$ is in km${}^2$ and $D$ in m. Cluster area $A=0.63$ ensures average number of BSs per cluster is $20$.}
\label{fig:CovPCSI}
\end{figure*}

The rate profile with limited CSI is provided for different configurations of cluster area, density of BSs and limited channel states, in   Table \ref{tab:rcdflim}. We  observe   that when  CSI limited to two  channels, the cell-edge rate improves from $0.12$ to $0.16$ and with  CSI   of  four and six channels, the cell-edge rate doubles and triples respectively. Similar observation hold true even for the mean rate. Also from the rate profile, we can see that it is the available channel state at the transmitter that matters more than the average number of nodes or area of the cluster. Hence, it is better to build small clusters of about six to eight BSs, so that the receiver complexity can be reduced, and still have performance benefits as compared to that of no clustering.
Hence, ZF-DPC is a good choice for cloud radio network operation in the downlink due to a high gain in capacity even with limited feedback. Furthermore, limiting the CSI to about eight channels captures about $75-81$\% of the rate that could be achieved with complete CSI from the cluster nodes.
\begin{table}
\centering
\begin{subtable}{.33\textwidth}
\centering
\begin{tabular}{l c c c c}
  \hline
$L$&5 \% & 10 \% &  50 \% & Mean\\
\hline
Cell& 0.12 & 0.20 & 1.30 &  2.38 \\
2& 0.16 & 0.31 & 1.91 &  2.83 \\
4& 0.27 & 0.50 & 2.64 &  3.50 \\
6& 0.32 & 0.59 & 3.03 &  3.86 \\
8& 0.35 & 0.65 & 3.27 &  4.09 \\
Full& 0.42 & 0.76 & 3.67 &  4.49 \\
\hline
\end{tabular}
\label{tab:rcdf_lim_200_16}
\caption{$D=200$, $A=0.5$, $\bar{N}=16$}
\end{subtable}%
\begin{subtable}{.33\textwidth}
\centering
\begin{tabular}{l c c c c}
  \hline
$L$&5 \% & 10 \% &  50 \% & Mean\\
\hline
Cell& 0.12 & 0.20 & 1.30 &  2.38 \\
2& 0.16 & 0.31 & 1.91 &  2.84 \\
4& 0.29 & 0.56 & 2.77 &  3.63 \\
6& 0.39 & 0.70 & 3.25 &  4.05 \\
8& 0.42 & 0.76 & 3.67 &  4.49 \\
Full& 0.54 & 1.02 & 4.48 &  5.13 \\
\hline
\end{tabular}
\label{tab:rcdf_lim_200_32}
\caption{$D=200$, $A=1$, $\bar{N}=32$}
\end{subtable}
\begin{subtable}{.33\textwidth}
\centering
\begin{tabular}{l c c c c}
  \hline
$L$&5 \% & 10 \% &  50 \% & Mean\\
\hline
Cell& 0.12 & 0.20 & 1.30 &  2.38 \\
2& 0.17 & 0.32 & 1.96 &  2.86 \\
4& 0.32 & 0.60 & 2.85 &  3.70 \\
6& 0.43 & 0.79 & 3.38 &  4.17 \\
8& 0.52 & 0.93 & 3.72 &  4.49 \\
Full& 0.71 & 1.30 & 5.28 &  5.80 \\
\hline
\end{tabular}
\label{tab:rcdf_lim_200_64}
\caption{$D=200$, $A=2$, $\bar{N}=64$}
\end{subtable}

\hfill\begin{subtable}{.33\textwidth}
\centering
\begin{tabular}{l c c c c}
  \hline
$L$&5 \% & 10 \% &  50 \% & Mean\\
\hline
Cell& 0.12 & 0.20 & 1.30 &  2.38 \\
2& 0.15 & 0.29 & 1.94 &  2.87 \\
4& 0.29 & 0.52 & 2.65 &  3.51 \\
6& 0.32 & 0.60 & 3.00 &  3.79 \\
8& 0.38 & 0.71 & 3.40 &  4.18 \\
Full& 0.39 & 0.75 & 3.67 &  4.48 \\
\hline
\end{tabular}
\label{tab:rcdf_lim_400_16}
\caption{$D=400$, $A=2$, $\bar{N}=16$}
\end{subtable}\hfill
\begin{subtable}{.33\textwidth}
\centering
\begin{tabular}{l c c c c}
  \hline
$L$&5 \% & 10 \% &  50 \% & Mean\\
\hline
Cell& 0.12 & 0.20 & 1.30 &  2.38 \\
2& 0.16 & 0.31 & 1.92 &  2.85 \\
4& 0.30 & 0.55 & 2.76 &  3.62 \\
6& 0.40 & 0.72 & 3.26 &  4.07 \\
8& 0.43 & 0.79 & 3.52 &  4.30 \\
Full& 0.58 & 1.04 & 4.54 &  5.21 \\
\hline
\end{tabular}
\label{tab:rcdf_lim_400_32}
\caption{$D=400$, $A=4$, $\bar{N}=32$}
\end{subtable}\hfill
\begin{subtable}{.33\textwidth}
\centering
\begin{tabular}{l c c c c}
  \hline
$L$&5 \% & 10 \% &  50 \% & Mean\\
\hline
Cell& 0.12 & 0.20 & 1.30 &  2.38 \\
2& 0.16 & 0.32 & 1.97 &  2.88 \\
4& 0.33 & 0.61 & 2.85 &  3.69 \\
6& 0.44 & 0.79 & 3.37 &  4.16 \\
8& 0.51 & 0.92 & 3.73 &  4.45 \\
Full& 0.72 & 1.35 & 5.45 &  5.95 \\
\hline
\end{tabular}
\label{tab:rcdf_lim_400_64}
\caption{$D=400$, $A=8$, $\bar{N}=64$}
\end{subtable}\hfill
\caption{Rate profile of a typical user for various values of PCSI,  cluster area and BS density. Rates are in bits/sec/Hz. \textit{Cell} and \textit{Full} represent conventional cellular networks without cooperation and clustering with full cooperation respectively.}
\label{tab:rcdflim}
\end{table}

\section{Conclusion}\label{sec:conc}
In this paper, the impact of finite clustering in a cloud radio on the coverage and rate is analyzed with complete and with partial channel feedback. In particular, the coverage probability of a typical user served by group of geographically clustered BSs is provided.    We compared the cloud network performance with a conventional cellular network. From this analysis and simulations, we observe that the average number of nodes and available channel state feed back in a cluster play a critical role, rather than the geographical area of a cooperating cluster. Even finite clustering provides substantial rate gain to both the cluster edge and interior users.  It is found that even with practically feasible feedback of  the nearest six channel states from each user, the edge rate can be quadrupled. 
 \section*{Acknowledgment}
 We would like to acknowledge the IU-ATC project for its support. This project is funded by the Department of Science and Technology (DST), India and Engineering and Physical
Sciences Research Council (EPSRC). We would also like to thank the CPS project in IIT Hyderabad and the Samsung GRO project for supporting this work.

\bibliographystyle{IEEEtran}
\bibliography{IEEEabrv,report2xN}

\appendices
\section{Proof of Lemma \ref{the:cov}}\label{sec:app_clust_full}
Here, for notational simplicity, $\lambda$, $r_b$, $h_b$ and $h_j$ will be used  instead of $\lambda_b$, $r_{ii}$, $h_{ii}$ and $h_{i,j}$ respectively. The users are classified into three types according to $r_b$, $r_u$ and $R$. Let $D_u=\mathcal{B}(u,r_b)$.

\begin{tabbing}
\hspace*{0.2cm}\=Type-I\hspace*{0.4cm}\= :\hspace*{0.2cm}\=\begin{minipage}[t]{0.85\linewidth}
A user is classified as type-I  when $D_c\cap D_u=D_u$, \ie,   $r_u+r_b\leq R$.  In this case, the interfering   BSs will belong to $\Phi_b\setminus D_c\cap \Phi_b$.  \end{minipage}\\[5pt]
\>Type-II\> : \>\begin{minipage}[t]{0.9\linewidth}
A user is classified as type-II when $R-r_u<r_b<R+r_u$.  In this case, the void region, $D_u$, around the user $u$ is not entirely in $D_c$. The interfering BSs are $\Phi_b\setminus (D_c\cup D_u \cap \Phi_b)$.
\end{minipage}\\[5pt]
\>Type-III\> : \>\begin{minipage}[t]{0.9\linewidth}
 A user is classified as type-III when $D_c\subset D_u$. However this is a very low probability event and happens only when the density of the BSs is very low and these users are not belonged to $C_0$ as these are tagged to BSs out side $C_0$.
\end{minipage}\\[5pt]
\end{tabbing}

The probabilities of type-I and type-II users are given by,
\begin{align*}
p_1&=\P[0<r_u<R,0<r_b<R-r_u],\\
&=1-\frac{ \text{erf}\left(\sqrt{\pi } \sqrt{\lambda } R\right)}{{\sqrt{\lambda }  R}}-\frac{e^{-\pi  \lambda
   R^2}-1}{\pi \lambda  R^2},
\end{align*} and similarly
\begin{align*}
p_2&=\frac{\pi R\sqrt{\lambda }\, \text{erf}\left(2 \sqrt{\pi } \sqrt{\lambda }
   R\right)+e^{-4 \pi  \lambda  R^2}-1}{\lambda\pi  R^2}.
\end{align*}

\begin{figure*}%
\centering
\hfill
\begin{subfigure}{.45\columnwidth}
\begin{tikzpicture}[scale=0.7]
\draw[fill=gray!20,draw=black!0] (-5,-5)--(5,-5) -- (5,5) -- (-5,5)-- cycle;
    \def\firstellipse{(0,0) circle (4)}
    \def\secondellipse{(2,1.5) circle (1.2)}
    \fill[gray!20] \firstellipse \secondellipse;

    \begin{scope}
        \fill[white] \firstellipse;
    \end{scope}

    \draw \firstellipse \secondellipse;
	\draw[] (-5,0) -- (5,0)  coordinate (x axis);
	\draw[] (0,-5) -- (0,5) coordinate (y axis);
	\node [fill,draw,scale=0.3,circle]at (2,1.5) {};
	\draw (2.3,1.8) node  {$u$};
	\draw[] (0,0) -- (2,1.5); 
	\draw (1.1,1.1) node  {$r_u$};
	\node [draw,scale=0.2,regular polygon,regular polygon sides=3,fill=black!100]at (3.1,1){};
	\draw[] (2,1.5) -- (3.1,1); 
	\draw (2.7,1.4) node  {$r_b$};
	\node [draw,scale=0.2,regular polygon,regular polygon sides=3,fill=black!100]at (0.4,3){};
	\node [draw,scale=0.2,regular polygon,regular polygon sides=3,fill=black!100]at (-1,-2){};
	\node [draw,scale=0.2,regular polygon,regular polygon sides=3,fill=black!100]at (1,0.4){};
	\node [draw,scale=0.2,regular polygon,regular polygon sides=3,fill=black!100]at (-3,-1.6){};
	\node [draw,scale=0.2,regular polygon,regular polygon sides=3,fill=black!100]at (3,-1.6){};
	\node [draw,scale=0.2,regular polygon,regular polygon sides=3,fill=black!100]at (4,1.8){};
	\node [draw,scale=0.2,regular polygon,regular polygon sides=3,fill=black!100]at (2,4.6){};
	\node [draw,scale=0.2,regular polygon,regular polygon sides=3,fill=black!100]at (-3,4.2){};
	\node [draw,scale=0.2,regular polygon,regular polygon sides=3,fill=black!100]at (-2,2.2){};
	\node [draw,scale=0.2,regular polygon,regular polygon sides=3,fill=black!100]at (-2,-4.2){};
\end{tikzpicture}
\caption[]{Type-I user}
\label{fig:TypeI}%
\end{subfigure}\hfill%
\begin{subfigure}{.45\columnwidth}
\centering
 \begin{tikzpicture}[scale=0.7]
 \draw[fill=gray!20,draw=black!0] (-5,-5)--(5,-5) -- (5,5) -- (-5,5)-- cycle;
 	\draw[fill=blue!20,draw=black!0] (1.3,-5)--(5,-5) -- (5,-1.8) -- (2,-2.2)-- cycle;
     \def\firstellipse{(0,0) circle (4)}
     \def\secondellipse{(2,-2.2) circle (1.5)}
     \fill[gray!20] \firstellipse \secondellipse;
     \begin{scope}
         \fill[white] \secondellipse;
         \fill[white] \firstellipse;
     \end{scope}
     \draw \firstellipse \secondellipse;
 	\draw[] (-5,0) -- (5,0)  coordinate (x axis);
 	\draw[] (0,-5) -- (0,5) coordinate (y axis);
 	\node [fill,draw,scale=0.3,circle]at (2,-2.2) {};
 	\draw (1.8,-2.4) node  {$u$};
 	\draw[] (0,0) -- (2,-2.2); 
 	\draw (1.5,-1.3) node  {$r_u$};
 	\draw[] (2,-2.2) -- (1.3,-5);
 	\draw[] (2,-2.2) -- (5,-1.8);
 	\draw[] (2,-2.2) -- (4.5,-5);
 	\draw[dashed] (2,-2.2) -- (1.2,-3.8);
 	\draw (1,-3.6) node  {$R^\prime$};
 	\node [draw,scale=0.2,regular polygon,regular polygon sides=4,fill=black!100]at (1.2,-3.8){};
 	\node [draw,scale=0.2,regular polygon,regular polygon sides=4,fill=black!100]at (3.6,-4){};
 	\draw (3.7,-3.8) node  {$A$};
 	\node [draw,scale=0.2,regular polygon,regular polygon sides=4,fill=black!100]at (3.5,-2){};	
 	\draw (3.7,-2.3) node  {$B$};
 	\draw (4.5,-3) node  {$\mathcal{R}_1$};
 	\draw (4.5,3) node  {$\mathcal{R}_2$};
 	\node [draw,scale=0.2,regular polygon,regular polygon sides=3,fill=black!100]at (3,-1.1){};
 	\draw[] (2,-2.2) -- (3,-1.1); 
 	\draw (2.5,-1.3) node  {$r_b$};
 	\node [draw,scale=0.2,regular polygon,regular polygon sides=3,fill=black!100]at (0.4,3){};
 	\node [draw,scale=0.2,regular polygon,regular polygon sides=3,fill=black!100]at (-1,-2){};
 	\node [draw,scale=0.2,regular polygon,regular polygon sides=3,fill=black!100]at (1,0.4){};
 	\node [draw,scale=0.2,regular polygon,regular polygon sides=3,fill=black!100]at (-3,-1.6){};
 	\node [draw,scale=0.2,regular polygon,regular polygon sides=3,fill=black!100]at (4,1.8){};
 	\node [draw,scale=0.2,regular polygon,regular polygon sides=3,fill=black!100]at (2,4.6){};
 	\node [draw,scale=0.2,regular polygon,regular polygon sides=3,fill=black!100]at (-3,4.2){};
 	\node [draw,scale=0.2,regular polygon,regular polygon sides=3,fill=black!100]at (-2,2.2){};
 	\node [draw,scale=0.2,regular polygon,regular polygon sides=3,fill=black!100]at (-2,-4.2){};
  \end{tikzpicture}
 \caption{Type-II user}%
 \label{subfigb}%
 \end{subfigure}\hfill
 \caption{Illustration of type-I and II users and their interference region. For type-I user interferer region is outside $D_C$ and for type-II it is outside $D_C\cup D_u$ also for type-II the region is divided into two, $\mathcal{R}_1$ and $\mathcal{R}_2$. In $\mathcal{R}_2$, interferer distance, $r_y\geq r_b$ and in $\mathcal{R}_1$ $r_y\geq R^\prime$.}
 \label{fig:TypeII}
 \end{figure*}
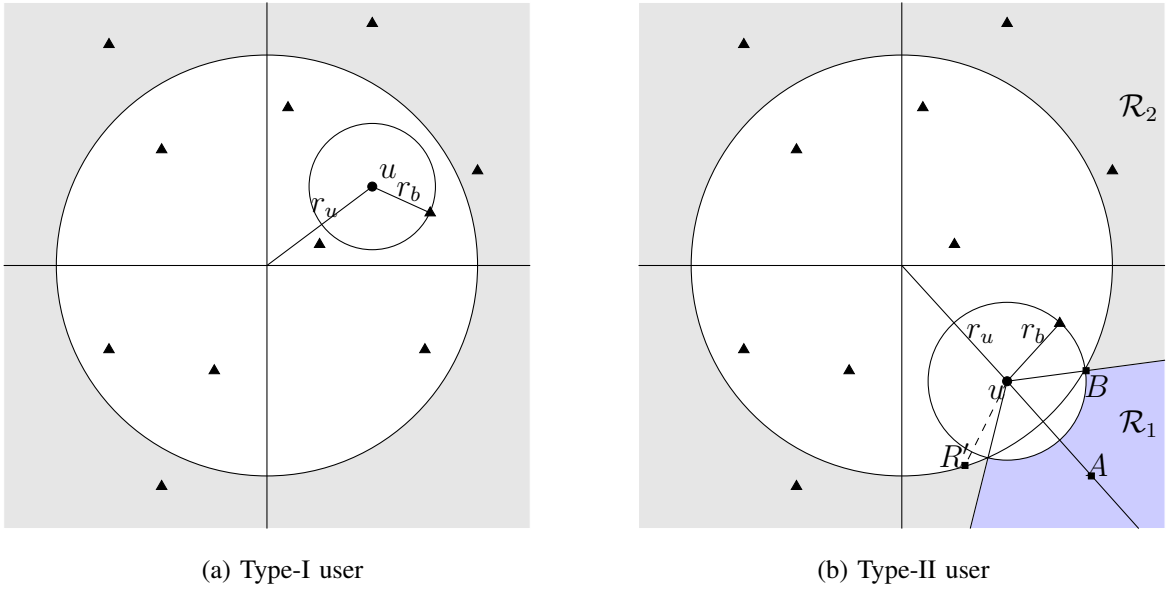

\subsection{User Type-I}
 By shifting origin to $u$, the $\sinr$ of the user at the origin  can be written as,
    \begin{align}
   \sinr&=\dfrac{|h|^2r_{b}^{-\alpha}}{\sigma^2+\sum_{y\in\Phi\setminus D_R}|h_y|^2 \|y \|^{-\alpha}},
    \end{align} where $D_R=\Phi\cap\mathcal{B}(-u,R)$.
The conditional distribution of $\sir$ is $\P[\sir>T|r_u,r_b]$
\begin{align}
&=\P\left[\dfrac{|h_{b}|^2r_{b}^{-\alpha}}{\sigma^2+\sum_{y\in\Phi\setminus D_R}|h_y|^2 \|y \|^{-\alpha}}>T|r_u,r_b\right],\nonumber\\
&=e^{-T r_b^\alpha \sigma^2}\L_I(T r_b^\alpha).\label{eq:cov_condit}
\end{align}
Laplace transform of interference,
\begin{align*}
\L_I(s)&= \E\left[ \exp\left(- s\sum_{y\in\Phi\setminus D_R}|h_y|^2 \|y\|^{-\alpha}\right)\right],\\
&=\E\left[ \prod_{y\in\Phi\setminus D_R}\dfrac{1}{1+s\|y \|^{-\alpha}}\right],\\
&=\exp\left(-\lambda \int_{\theta=0}^{2\pi}\int_{R^\prime}^\infty \left(1-\dfrac{1}{1+s r_y^{-\alpha}}\right)r_y\d r_y \d \theta \right),\\
&=\exp\left(-\lambda \int_{\theta=0}^{2\pi}\frac{s (R^\prime)^{2-\alpha } \, _2F_1\left(1,\frac{\alpha -2}{\alpha
   };2-\frac{2}{\alpha };-{(R^\prime)}^{-\alpha } s\right)}{\alpha -2} \d \theta \right)\label{eq:laplace_exact1}
\end{align*} where $R^\prime=\sqrt{\cos ^2(\theta ) r_u^2+R^2}-\cos (\theta ) r_u$ is the distance of $u$ to the point on the disc $D_c$ at an angle $\theta$.
\subsection{User Type-II}
For users with $0<r_u<R$, $R-r_u<r_b<R+r_u$,
\begin{align}
   \sinr&=\dfrac{|h_{b}|^2r_{b}^{-\alpha}}{\sigma^2+\sum_{y\in\Phi\setminus D_E}|h_y|^2 \|y-u \|^{-\alpha}},
    \end{align} where
    $D_E=\Phi\cap(\mathcal{B}(0,R)\cup\mathcal{B}(u,r_b))$.

 The interfering regions are shaded in Fig.~\ref{fig:TypeII}. In the region, $\mathcal{R}_1$, the interferer distance, $r_y$, from user is such that $r_b<r_y$ and in region, $\mathcal{R}_2$, interferes are beyond the disc perimeter.
 The distance from any $u\in D_c$ to a point on the perimeter of the disc at an angle $\varphi$ from $UA$ is given by $y=R^\prime=\sqrt{\cos^2(\varphi) r_u^2+R^2}-\cos^2(\varphi) r_u$.
 Laplace transform of interference,
 \begin{align*}
 \L_I(s)&= \E\left[ \exp\left(- s\sum_{y\in\Phi\setminus D_E}|h_y|^2 \|y\|^{-\alpha}\right)\right]\\
 &= \exp\left(-\lambda\int_{\varphi=-\Theta}^{\Theta}\int_{r_y=r_b}^{\infty}\left(1-\frac{1}{1+s r_y^{-\alpha}} \right)r_y\d r_y\d \varphi\right.\nonumber\\
 &\left. -\lambda\int_{\varphi=\Theta}^{2\pi-\Theta}\int_{r_y=R^\prime}^{\infty}\left(1-\frac{1}{1+s r_y^{-\alpha}} \right)r_y\d r_y\d \varphi\right),\\
 &=e^{-\lambda 2 \Theta  \frac{ s r_b^{2-\alpha} \, _2F_1\left(1,\frac{\alpha -2}{\alpha
     };2-\frac{2}{\alpha };-T \right)}{\alpha -2}}\nonumber\\
     &\times e^{-\lambda\int_{-\Theta}^{\Theta}\frac{s \left({R^\prime}\right){}^{2-\alpha } \,
        _2F_1\left(1,\frac{\alpha -2}{\alpha };2-\frac{2}{\alpha };-T r_b^\alpha \left({R^\prime}\right){}^{-\alpha }\right)}{\alpha -2}\d \varphi},\label{eq:laplace_exact2}
 \end{align*}
 where $\Theta$ is the angle between $u$ and the intersection of the circles, $\angle AUB$, \ie,  $\Theta=\pi-\arccos\left[ \frac{r_b^2+r_u^2-R^2}{2 r_b r_u}\right]$ and $R^\prime=\sqrt{b_n^2 r_u^2-r_u^2+R^2}-b_n r_u$ is the distance from origin to the arc of cluster disc from $\Theta$ to $2\pi-\Theta$.
\section{Proof of lemma \ref{lem:scaling}}\label{sec:app_scaling}
\begin{proof}
	The probability that the BS that the user at $y=(R(1-\delta),0)$ associates with a BS in the cluster is lower bounded by
	\[\P(\Phi(B(y,\delta))>0)=1-\exp(-\lambda R^2\delta^2),\]
	and tends to one as $R$ becomes larger.  If $\epsilon \approx 0$, this implies that the cluster radius $R$ should be very large. Hence the interference (with very high probability) is from all the BSs outside the cluster. We now compute   the probabilities conditioned on the distance $r$ of the user to the closest BS. Since $|h|^2$ is exponentially distributed, the conditional probability is 
	\[{\tilde{P}}_{c,R}(T)=\exp(-T\sigma^2 r^\alpha)\underbrace{\E[\exp(-T r^\alpha \I_i(\Phi))]}_{T_1}.\]
	We will now look at term containing the interference.  First averaging over the fading terms in the interference which are exponentially distributed,
	\begin{align*}
	T_1=& \E\prod_{x\in \Phi\cap B(o,R)^c}e^{-T r^\alpha h_{xy}\|x-y\|^{-\alpha}},\\
	=&\E\prod_{x\in \Phi\cap B(o,R)^c}\frac{1}{1+T r^\alpha\|x-y\|^{-\alpha}}.
	\end{align*}
	Using the probability generating functional of the Poisson point process and using the substitution $x\to z'R$ we have
	\begin{align}
	T_1 = e^{-\lambda R^2\int_{B(o,1)^c} \frac{1}{1+T^{-1} r^{-\alpha}R^\alpha\|z-\tilde{y}\|^{\alpha}}\d z},
	\label{eq:act}
	\end{align}
	where  $\tilde{y} =(1-\delta,0)$.
	We want $T_1$ to be approximately $1-\epsilon$. When $\epsilon$ is small it is easy to see that the argument of the exponent should be close to zero. This is possible only when $R$ is very large, and for large $R$ we have
	\begin{align}
	T_1 \sim e^{-\lambda  R^{2-\alpha}T r^\alpha\int_{B(o,1)^c}^\infty \frac{1}{\|z-\tilde{y}\|^{\alpha}}\d z}.
	\end{align}
	Let  $\Delta(\delta) =\int_{B(o,1)^c}^\infty \frac{1}{\|z-\tilde{y}\|^{\alpha}}\d z$. Observe that $\Delta(\delta)<\infty$ for $\delta<1$. Hence averaging over the radius $r$, 
	%
	%
	\begin{align}
	P_{c,R}(T) \sim&  \E_r\left[\exp(-T\sigma^2 r^\alpha)  e^{-\lambda  R^{2-\alpha}T r^\alpha\Delta(\delta)}\right]\\
	\stackrel{(a)}{\sim}& \E_r\left[\exp(-T\sigma^2 r^\alpha) (1- \lambda  R^{2-\alpha}T r^\alpha\Delta(\delta))\right],
	\label{eq:asym}
	\end{align}
	where $(a)$ follows by the approximation $\exp(-x) \sim 1-x, x \to 0$. Setting 
	\begin{align*}
	R= \left(\frac{\lambda  T \Delta(\delta)\eta(T)}{\ln(1/(1-\epsilon))}\right)^{1/(\alpha-2)} \sim\left( \lambda  T \Delta(\delta)\eta(T) \epsilon^{-1}\right)^{1/(\alpha-2)},
	\end{align*}
	where $\eta(T) =\E[\exp(-T\sigma^2 r^\alpha)r^\alpha]/\E[\exp(-T\sigma^2 r^\alpha) ]$ would lead to 
	\[P_{c,R}(T) \sim (1-\epsilon)\E_r \left[\exp(-T\sigma^2 r^\alpha) \right].\]
\end{proof}
 \section{Proof of lemma \ref{the:cov_lim}}\label{sec:app_clust_lim}
 Here, the users are classified into five types according to $r_b$, $r_u$, $R$ and $r_l$, where $r_l$ is the distance to $L$th nearest interferer from the typical user. Let $D_u=\mathcal{B}(u,r_b)$ and $D_L=\mathcal{B}(u,r_l)$. The possible five cases of $r_u$, $r_b$ and $r_l$ are shown in Fig.~\ref{fig:CloudClustPCSI}.

\begin{tabbing}
\hspace*{0.2cm}\=Type-I\hspace*{0.4cm}\= :\hspace*{0.2cm}\=\begin{minipage}[t]{0.85\linewidth}
A user is classified as type-I  when $D_c\cap D_u=D_u$, \ie,   $r_u+r_b\leq R$ and $D_c\cap D_L=D_L$, \ie,    $r_b\leq r_l \leq R-r_u$.  In this case, the interfering   BSs will belong to $\Phi_b\setminus D_L\cap \Phi_b$, see Fig.~\ref{fig:PCSITypeI}.  Probability of this users will be high for dense networks. \end{minipage}\\[5pt]
\>Type-II\> : \>\begin{minipage}[t]{0.9\linewidth}
Here, $0\leq r_b<R-r_u$ and $R-r_u\leq r_l \leq R+r_u$. Here the void region $D_u$ is inside $D_c$ and the interfering BSs belongs to $\Phi_b\setminus (D_c\cap D_L \cap \Phi_b)$, see Fig.~\ref{fig:PCSITypeII}.
\end{minipage}\\[5pt]
\>Type-III\> : \>\begin{minipage}[t]{0.9\linewidth}
Here, $R-r_u\leq r_b<R+r_u$ and $r_b\leq r_l \leq R+r_u$. Here the void region $D_u$ is not perfectly inside $D_c$ and the interfering BSs belongs to $\Phi_b\setminus (D_c\cap D_L\cup D_u \cap \Phi_b)$, see Fig.~\ref{fig:PCSITypeIII}.
\end{minipage}\\[5pt]
\>Type-IV\> : \>\begin{minipage}[t]{0.9\linewidth}
When $D_L\cap D_c=D_L$, means the number of interfering BSs inside the cluster of interest is less than $L$, we have not considered this case.
\end{minipage}\\[5pt]
\>Type-V\> : \>\begin{minipage}[t]{0.9\linewidth}
Physically inside the cluster but not tagged to a BS of the cluster area it belongs to.
\end{minipage}\\[5pt]
\end{tabbing}

\begin{figure}%
\hfill\begin{subfigure}{.33\columnwidth}
\centering
\begin{tikzpicture}[scale=0.5]
	\draw[fill=gray!20,draw=black!0] (-5,-5)--(5,-5) -- (5,5) -- (-5,5)-- cycle;
    \def\firstellipse{(0,0) circle (4)}
    \def\secondellipse{(-0.5,0.4) circle (1.5)}
    \def\thirdellipse{(-0.5,0.4) circle (3.2)}
    \fill[gray!20] \firstellipse \secondellipse;

    \begin{scope}
        \clip \firstellipse;
        \fill[white] \secondellipse;
        \fill[white] \thirdellipse;
    \end{scope}

    \draw \firstellipse \secondellipse \thirdellipse;
	\draw[] (-5,0) -- (5,0)  coordinate (x axis);
	\draw[] (0,-5) -- (0,5) coordinate (y axis);
	\node [fill,draw,scale=0.3,circle]at (-0.5,0.4) {};
	\draw (-0.8,0.8) node  {$u$};
	\draw[] (-0.5,0.4) -- (1,0.4); 
	\draw (0.2,0.7) node  {$r_b$};
	\node [draw,scale=0.2,regular polygon,regular polygon sides=3,fill=black!100]at (1,0.4){};
	\node [draw,scale=0.2,regular polygon,regular polygon sides=3,fill=black!100]at (2,1.4){};
	\node [draw,scale=0.3,regular polygon,regular polygon sides=3,fill=black!100]at (-3,-1.6){};
	\draw[] (-0.5,0.4) -- (-3,-1.6); 
	\node [draw,scale=0.2,regular polygon,regular polygon sides=3,fill=black!100]at (3,-1.6){};
	\node [draw,scale=0.2,regular polygon,regular polygon sides=3,fill=black!100]at (4,1.8){};
	\node [draw,scale=0.2,regular polygon,regular polygon sides=3,fill=black!100]at (2,4.6){};
	\node [draw,scale=0.2,regular polygon,regular polygon sides=3,fill=black!100]at (-3,4.2){};
	\node [draw,scale=0.2,regular polygon,regular polygon sides=3,fill=black!100]at (-2,2.2){};
	\node [draw,scale=0.2,regular polygon,regular polygon sides=3,fill=black!100]at (-2,-4.2){};
\end{tikzpicture}
\caption[]{Type-I}
\label{fig:PCSITypeI}%
\end{subfigure}\hfill%
\begin{subfigure}{.33\columnwidth}
\centering
\begin{tikzpicture}[scale=0.5]
	\draw[fill=gray!20,draw=black!0] (-5,-5)--(5,-5) -- (5,5) -- (-5,5)-- cycle;
	    \def\firstellipse{(0,0) circle (4)}
	    \def\secondellipse{(2,1.5) circle (1.2)}
	    \def\thirdellipse{(2,1.5) circle (2.5)}
	    \fill[gray!20] \firstellipse \secondellipse;
	
	    \begin{scope}
	        \clip \firstellipse;
	        \fill[white] \secondellipse;
	        \fill[white] \thirdellipse;
	    \end{scope}
	
	    \draw \firstellipse \secondellipse \thirdellipse;
		\draw[] (-5,0) -- (5,0)  coordinate (x axis);
		\draw[] (0,-5) -- (0,5) coordinate (y axis);
		\node [fill,draw,scale=0.3,circle]at (2,1.5) {};
		\draw (2.1,1.9) node  {$u$};
		\draw[] (0,0) -- (2,1.5); 
		\draw (1.1,1.1) node  {$r_u$};
		\node [draw,scale=0.2,regular polygon,regular polygon sides=3,fill=black!100]at (3.1,1){};
		\draw[] (2,1.5) -- (3.1,1); 
		\draw (2.7,1.4) node  {$r_b$};
		\node [draw,scale=0.2,regular polygon,regular polygon sides=3,fill=black!100]at (0.4,3){};
		\node [draw,scale=0.2,regular polygon,regular polygon sides=3,fill=black!100]at (1,0.4){};
		\node [draw,scale=0.2,regular polygon,regular polygon sides=3,fill=black!100]at (-2,-2){};
		\node [draw,scale=0.2,regular polygon,regular polygon sides=3,fill=black!100]at (3,-1.6){};
		\node [draw,scale=0.2,regular polygon,regular polygon sides=3,fill=black!100]at (4,1.8){};
		\node [draw,scale=0.2,regular polygon,regular polygon sides=3,fill=black!100]at (2,4.6){};
		\node [draw,scale=0.2,regular polygon,regular polygon sides=3,fill=black!100]at (-3,4.2){};
		\node [draw,scale=0.2,regular polygon,regular polygon sides=3,fill=black!100]at (-2,2.2){};
		\node [draw,scale=0.2,regular polygon,regular polygon sides=3,fill=black!100]at (-2,-4.2){};
\end{tikzpicture}
 \caption{Type-II}%
 \label{fig:PCSITypeII}%
 \end{subfigure} \hfill%
 \begin{subfigure}{.33\columnwidth}
 \centering
\begin{tikzpicture}[scale=0.5]
	\draw[fill=gray!20,draw=black!0] (-5,-5)--(5,-5) -- (5,5) -- (-5,5)-- cycle;
	    \def\firstellipse{(0,0) circle (4)}
	    \def\secondellipse{(2,-2.2) circle (1.5)}
	    \def\thirdellipse{(2,-2.2) circle (2.5)}
	    \fill[gray!20] \firstellipse \secondellipse;
	
	    \begin{scope}
	        \clip \thirdellipse;
	        \fill[white] \secondellipse;
	        \fill[white] \firstellipse;
	    \end{scope}
	
	    \draw \firstellipse \secondellipse \thirdellipse;
		\draw[] (-5,0) -- (5,0)  coordinate (x axis);
		\draw[] (0,-5) -- (0,5) coordinate (y axis);
		\node [fill,draw,scale=0.3,circle]at (2,-2.2) {};
		\draw (1.8,-2.4) node  {$u$};
		\draw[] (0,0) -- (2,-2.2); 
		\draw (1.5,-1.3) node  {$r_u$};
		\node [draw,scale=0.2,regular polygon,regular polygon sides=3,fill=black!100]at (3,-1.1){};
		\draw[] (2,-2.2) -- (3,-1.1); 
		\draw (2.5,-1.3) node  {$r_b$};
		\node [draw,scale=0.2,regular polygon,regular polygon sides=3,fill=black!100]at (0.4,3){};
		\node [draw,scale=0.2,regular polygon,regular polygon sides=3,fill=black!100]at (-1,-2){};
		\node [draw,scale=0.2,regular polygon,regular polygon sides=3,fill=black!100]at (1,0.4){};
		\node [draw,scale=0.2,regular polygon,regular polygon sides=3,fill=black!100]at (-3,-1.6){};
		\node [draw,scale=0.2,regular polygon,regular polygon sides=3,fill=black!100]at (4,1.8){};
		\node [draw,scale=0.2,regular polygon,regular polygon sides=3,fill=black!100]at (2,4.6){};
		\node [draw,scale=0.2,regular polygon,regular polygon sides=3,fill=black!100]at (-3,4.2){};
		\node [draw,scale=0.2,regular polygon,regular polygon sides=3,fill=black!100]at (-2,2.2){};
		\node [draw,scale=0.2,regular polygon,regular polygon sides=3,fill=black!100]at (-2,-4.2){};
\end{tikzpicture}
  \caption{Type-III}%
  \label{fig:PCSITypeIII}%
  \end{subfigure}\hfill
 \caption{Illustration of type-I,II and III users and their interference region. The shaded region is the possible locations of interference. For type-III user no interference region goes beyond the cluster disc because of the nearest neighbor connectivity which implies no interferer is closer than tagged BS.}
 \label{fig:CloudClustPCSI}
 \end{figure}
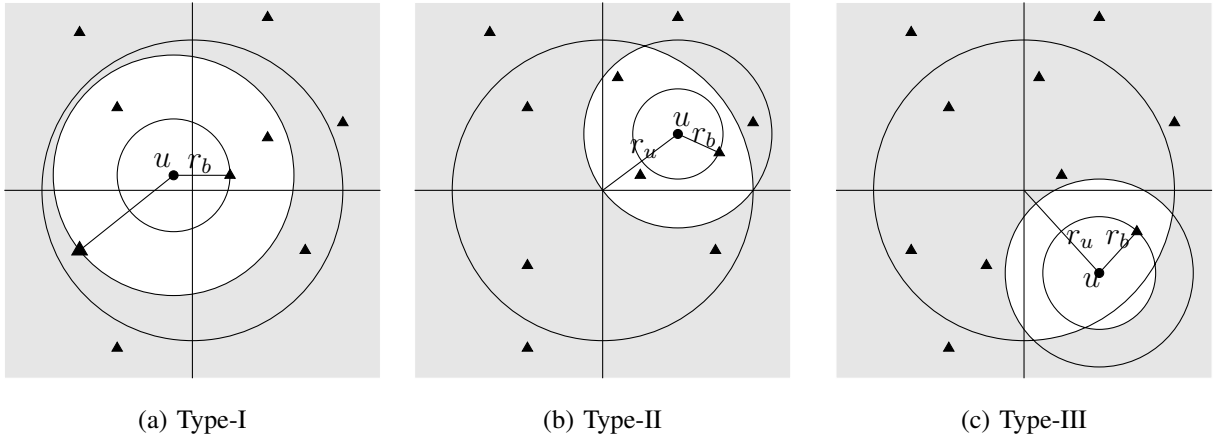
%

Last two cases can be neglected in the derivation, because they are not belonging to the cases we are analyzing here.
\subsubsection{Type I User}
The users who are finding the $L$th nearest BS is inside the cluster of the tagged BS are included in this case. Here, $0<r_u<R$, $0\leq r_b<R-r_u$ , $r_b\leq r_l \leq R-r_u$ and this assumption makes the user always connected to a BS inside its cluster and the interfering BSs will be the entire region outside the circle of radius $r_l$.
 Now the $\sinr$ of the user $u$ can be written as,
     \begin{align}
    \sinr&=\dfrac{|h_{b}|^2r_{b}^{-\alpha}}{\sigma^2+\sum_{y\in\Phi\setminus D_L}|h_y|^2 \|y-u \|^{-\alpha}}.
     \end{align}
 By shifting origin to $\us$,
     \begin{align}
    \sinr&=\dfrac{|h_{b}|^2r_{b}^{-\alpha}}{\sigma^2+\sum_{y\in\Phi\setminus \hat{D_L}}|h_y|^2 \|y \|^{-\alpha}}.
     \end{align}
      Therefore the conditional distribution of $\sir$ is,
 \begin{align}
 \P[\sir>T|r_u,r_b,r_l]&=\P\left[\dfrac{|h_{b}|^2r_{b}^{-\alpha}}{\sum_{y\in\Phi\setminus \hat{D_L}}|h_y|^2 \|y \|^{-\alpha}}>T|r_u,r_b,r_l\right],\\
 &=\P\left[|h_{b}|^2>Tr_b^\alpha\sum_{y\in\Phi\setminus \hat{D_L}}|h_y|^2 \|y\|^{-\alpha}|r_u,r_b,r_l\right],\\
 &=\L_I(T r_b^\alpha)|_{r_u,r_b,r_l}.
 \end{align}
 Laplace transform of interference,
 \begin{align}
 \L_I(s)&=\exp\left({-2\pi\lambda\frac{s r_l^{2-\alpha } \, _2F_1\left(1,\frac{\alpha -2}{\alpha };2-\frac{2}{\alpha };-s r_l^{-\alpha }\right)}{\alpha -2}}\right)
 \label{eq:laplace_exact3}
 \end{align}
\subsubsection{Type II User}
When the user finds the tagged BS inside the cluster and the $L$th BS distance $r_l$ is as given in ~Fig.\ref{fig:CloudClustPCSI},  \ie,   $D_u$, inside and $L$th interferer circle,  $D_L$  crossing the cluster area. Here, $0<r_u<R$, $0\leq r_b<R-r_u$ and $R-r_u\leq r_l \leq R+r_u$.

Laplace transform of interference,
\begin{align}
\L_I(s)&= \exp\left(-\lambda\int_{\varphi=-\Theta}^{\Theta}\int_{r_y=R^\prime}^{\infty}\left(1-\frac{1}{1+s r_y^{-\alpha}} \right)r_y\d r_y\d \varphi\right.\nonumber\\
&\left. -\lambda\int_{\varphi=\Theta}^{2\pi-\Theta}\int_{r_y=r_l}^{\infty}\left(1-\frac{1}{1+s r_y^{-\alpha}} \right)r_y\d r_y\d \varphi\right)\nonumber\\
&= \exp\left(-\lambda\int_{\varphi=-\Theta}^{\Theta}\frac{s (R^\prime)^{2-\alpha } \, _2F_1\left(1,\frac{\alpha -2}{\alpha
   };2-\frac{2}{\alpha };-(R^\prime)^{-\alpha } s\right)}{\alpha -2}\d \varphi\right)\nonumber\\
&\times \exp\left( -\lambda 2(\pi-\Theta)\frac{s r_l^{2-\alpha } \, _2F_1\left(1,\frac{\alpha
   -2}{\alpha };2-\frac{2}{\alpha };-s r_l^{-\alpha
   }\right)}{\alpha -2}\right)\label{eq:laplace_exact4}
\end{align}
where $\Theta$ is the angle between $\u$ and the point of intersection of the circles,  $\Theta=\pi-\arccos\left[ \frac{r_l^2+r_u^2-R^2}{2 r_l r_u}\right]$ and $R^\prime=\sqrt{b_n^2 r_u^2-r_u^2+R^2}-b_n r_u$ is the distance from origin to the arc of cluster disc from $\Theta$ to $2\pi-\Theta$. Also we have, $b_n=\cos(\varphi)$.


\subsubsection{Type III User}
When the user finds the tagged BS inside the cluster and the $L$th BS distance $r_l$ is as given in ~Fig.\ref{fig:CloudClustPCSI}, \ie,   $D_U$ and  $D_L$ are intersecting the cluster disc. Here, $0<r_u<R$, $R-r_u\leq r_b<R+r_u$ and $r_b\leq r_l \leq R+r_u$.

Laplace transform of interference,
\begin{align}
\L_I(s)&= \exp\left(-2\lambda\int_{\varphi=\Theta_1}^{\Theta_2}\int_{r_y=R^\prime}^{\infty}\left(1-\frac{1}{1+s r_y^{-\alpha}} \right)r_y\d r_y\d \varphi\right.\nonumber\\
&\left. -\lambda\int_{\varphi=-\Theta_1}^{\Theta_1}\int_{r_y=r_b}^{\infty}\left(1-\frac{1}{1+s r_y^{-\alpha}} \right)r_y\d r_y\d \varphi\right.\nonumber\\
&\left. -\lambda\int_{\varphi=\Theta_2}^{2\pi-\Theta_2}\int_{r_y=r_l}^{\infty}\left(1-\frac{1}{1+s r_y^{-\alpha}} \right)r_y\d r_y\d \varphi\right)\label{eq:laplace_exact5},
\end{align}
where $\Theta_1$ is the angle between $\mathbf{u}$ and the point of intersection of the circles $D_c$ and $D_u$, \ie,  $\Theta_1=\pi-\arccos\left[ \frac{r_b^2+r_u^2-R^2}{2 r_b r_u}\right]$, $\Theta_2$ is the angle between $\mathbf{u}$ and the point of intersection of the circles $D_c$ and $D_L$, \ie,  $\Theta_2=\pi-\arccos\left[ \frac{r_l^2+r_u^2-R^2}{2 r_l r_u}\right]$ and $R^\prime=\sqrt{b_n^2 r_u^2-r_u^2+R^2}-b_n r_u$ is the distance from origin to the arc of cluster disc from $\Theta$ to $2\pi-\Theta$. Also we have, $b_n=\cos(\varphi)$.


 \end{document}